\documentclass[journal,comsoc]{IEEEtran}

\usepackage{enumitem}
\usepackage{cite}
\usepackage{amsmath}
\usepackage{amssymb}
\usepackage{amsthm}
\usepackage{amsxtra}
\usepackage{amsfonts}
\usepackage{comment}
\usepackage{mathtools}
\usepackage{algorithmic}
\usepackage{graphicx}
\usepackage{xcolor}
\usepackage[nolist]{acronym}
\usepackage{subfig}
\usepackage{dsfont}
\usepackage{comment}
\usepackage{booktabs}
\usepackage{aligned-overset}
\usepackage[abbreviations,symbols,style=index,nonumberlist]{glossaries-extra}
\usepackage[normalem]{ulem}

\def\BibTeX{{\rm B\kern-.05em{\sc i\kern-.025em b}\kern-.08em
    T\kern-.1667em\lower.7ex\hbox{E}\kern-.125emX}}
\newtheorem{lemma}{Lemma}

\newtheorem{theorem}{Theorem}

\newtheorem{approximation}{Approximation}

\newcommand{\natZero}{\mathbb{N}_0}
\newcommand{\natOne}{\mathbb{N}}

\newcommand{\nodesSystem}{V} 
\newcommand{\frameSize}{m} 

\newcommand{\variableX}{x}
\newcommand{\variableT}{t}

\newcommand{\frameIndex}[1]{k(#1)} 
\newcommand{\position}[1]{\tau(#1)}

\newcommand{\nodeIndexF}{v}
\newcommand{\nodeIndexS}{u}
\newcommand{\process}{\Phi}

\newcommand{\indexN}{n}
\newcommand{\indexI}{i}
\newcommand{\indexJ}{j}

\newcommand{\selectedSlotVector}[1]{\underline{D}(#1)}
\newcommand{\selectedSlot}[2]{D^{#1}(#2)}
\newcommand{\rSelectedSlot}{d}

\newcommand{\frameState}[1]{\Lambda(#1)}
\newcommand{\AoI}[1]{\Delta(#1)}
\newcommand{\collidedFrames}[1]{C(#1)}
\newcommand{\OngoingTransmissions}[1]{B(#1)}
\newcommand{\FrameRecursive}[2]{\xi_{#1}(#2)}
\newcommand{\reservationStart}[2]{\gamma_{#1}(#2)}
\newcommand{\collidedReselections}[1]{W(#1)}

\newcommand{\selection}{s}
\newcommand{\Uniform}[1]{\mathrm{Uniform}\{#1\}}

\newcommand{\pReservationEnd}{p_\text{E}}

\newcommand{\setConnection}[1]{\mathcal{S}(#1)}

\newcommand{\threshold}{\theta}

\newcommand{\probM}[1]{\mathbb P\!\left(#1\right)}
\newcommand{\aAoI}{\bar{\Delta}}

\newcommand{\reselectingNodes}[1]{V_\text{R}(#1)}
\newcommand{\rReselectingNodes}{\bar{V}_\text{R}}

\newcommand{\rFrameState}{\lambda}

\newcommand{\emptySlots}[1]{\mathcal{N}(#1)}
\newcommand{\numEmptySlots}[1]{N(#1)}
\newcommand{\rNumEmptySlots}{n}
\newcommand{\probNumEmptySlots}[1]{P_N(#1)}
\newcommand{\eNumEmptySlots}{\mathbb{E}[N]}

\newcommand{\pRC}{p_\text{RC}}
\newcommand{\pKeep}{p_\text{Keep}}

\newcommand{\binomial}[3]{\mathrm{Bin}(#1,#2;#3)}

\newcommand{\rAoI}{\delta}
\newcommand{\eAoI}{\mathbb{E}[\bar{\Delta}]}

\newcommand{\rCollidedFrames}{c}

\newcommand{\rOngoingTransmissions}{b}
\newcommand{\maxOngoingTransmissions}{\Bar{b}}

\newcommand{\reservationLength}[2]{R^{#1}(#2)}
\newcommand{\rReservationLength}{r}
\newcommand{\setCollision}[2]{\mathcal{S}^{#1}_{#2}}
\newcommand{\rSetCollision}{\alpha}

\newcommand{\frameStateC}[2]{\Lambda_{#2}^{#1}}
\newcommand{\fRC}{\xi}

\newcommand{\ongoingTransmissionsC}[2]{B_{#2}^{#1}}

\newcommand{\permutation}{\pi}
\newcommand{\permSub}[2]{{#1}^{#2}_\pi}

\newcommand{\stateD}[1]{d^{#1}}
\newcommand{\stateC}[1]{c^{#1}}
\newcommand{\cOpp}[1]{c_{#1}}
\newcommand{\setOpp}{\mathcal{V}}

\newcommand{\rCollidedReselections}{w}

\newcommand{\maxCollidedReselections}{\bar{w}}

\newcommand{\functionQ}[1]{q(#1)}
\newcommand{\functionP}[1]{p(#1)}

\newcommand{\bigast}{\mathop{\begin{array}{@{}c@{}}%
 \text{\Huge$\ast$}\\[-0.9ex]\end{array}}}

\newcommand{\variationDistance}[1]{\eta(#1)}
\newcommand{\violation}[1]{\zeta_{#1}}

\begin{document}

\title{An Age of Information Characterization of SPS
\footnote{Copyright (c) 2026 IEEE. Personal use of this material is permitted. However, permission to use this material for any other purposes must be obtained from the IEEE by sending a request to pubs-permissions@ieee.org.}
\thanks{The authors would like to thank the Federal Ministry of Research, Technology, and Space (BMFTR) for its support of the project xG-RIC as part of the research program Communication Systems "Souverän. Digital. Vernetzt." (grant numbers 16KIS2429K and 16KIS2434).\\
Additionally, the authors acknowledge the use of AI tools to improve the readability of this paper.}}

\author{\IEEEauthorblockN{Maria Bezmenov\IEEEauthorrefmark{1}\IEEEauthorrefmark{2}, 
Matthias Frey\IEEEauthorrefmark{3},
Zoran Utkovski\IEEEauthorrefmark{4} and
Slawomir Stanczak \IEEEauthorrefmark{1}\IEEEauthorrefmark{4}}\\
\IEEEauthorblockA{\IEEEauthorrefmark{1}Technische Universität Berlin, Berlin, Germany},
\IEEEauthorblockA{\IEEEauthorrefmark{2}Robert Bosch GmbH, Hildesheim, Germany},\\
\IEEEauthorblockA{\IEEEauthorrefmark{3}University of Melbourne, Melbourne, Australia},
\IEEEauthorblockA{\IEEEauthorrefmark{4}Fraunhofer Heinrich-Hertz-Institute, Berlin, Germany}\\
Email:\{maria.bezmenov\}@campus.tu-berlin.de,
\{matthias.frey\}@unimelb.edu.au,\\
\{zoran.utkovski, slawomir.stanczak\}@hhi.fraunhofer.de}
\maketitle

\begin{abstract}
    We derive a closed-form approximation of the stationary distribution of the Age of Information (AoI) of the semi-persistent scheduling (SPS) protocol which is a core part of NR-V2X, an important standard for vehicular communications.
    While prior works have studied the average AoI under similar assumptions, in this work we provide a full statistical characterization of the AoI by deriving an approximation of its probability mass function. 
    As result, besides the average AoI, we are able to evaluate the age-violation probability, which is of particular relevance for safety-critical applications in vehicular domains, where the priority is to ensure that the AoI does not exceed a predefined threshold during system operation.     
    The study reveals complementary behavior of the age-violation probability compared to the average AoI and highlights the role of the duration of the reservation as a key parameter in the SPS protocol. We use this to demonstrate how this crucial parameter should be tuned according to the performance requirements of the application.
\end{abstract}

\begin{acronym}
    \acro{pmf}[pmf]{probability mass function}
    \acro{cdf}[cdf]{cumulative distribution function}
    \acro{AoI}[AoI]{Age of Information}
    \acro{PAoI}[PAoI]{Peak AoI}
    \acro{3GPP}[3GPP]{3rd generation partnership project}
    \acro{C-V2X}[C-V2X]{Cellular Vehicle-to-Everything}
    \acro{V2X}{Vehicle-to-Anything}
    \acro{TB}[TB]{Transport Block}
    \acro{SCI}[SCI]{Sidelink Control Information}
    \acro{RB}[RB]{Resource Block}
    \acro{MCS}[MCS]{Modulation and Coding Scheme} 
    \acro{RRI}[RRI]{Resource Reservation Interval}
    \acro{SPS}[SPS]{Semi-Persistent Scheduling}
    \acro{RC}[RC]{SL Resource Reselection Counter}
    \acro{MAC}[MAC]{Medium Access Control}
    \acro{PDR}[PDR]{Packet Delivery Ratio}
    \acro{UE}[UE]{User Equipment}
    \acro{OWD}[OWD]{One Way Delay}
    \acro{PIR}[PIR]{Packet Inter-Reception Time}
    \acro{PLR}[PLR]{Packet Loss Ratio}
    \acro{NR-V2X}[NR-V2X]{New Radio V2X}
    \acro{LTE-V2X}[LTE-V2X]{Long Term Evolution V2X}
    \acro{V2X}[V2X]{Vehicle-to-Everything}
    \acro{ITS}[ITS]{Intelligent Transportation Systems}
    \acro{CAM}[CAM]{Cooperative Awareness Message}
    \acro{CPM}[CPM]{Collective Perception Message}
    \acro{MCM}[MCM]{Maneuver Coordination Message}
    \acro{QoS}[QoS]{Quality of Service}
    \acro{RSRP}[RSRP]{Reference Signal Received Power}
    \acro{TBS}[TBS]{Transport Block Size}
    \acro{HD}[HD]{Half Duplex}
    \acro{IRSA}[IRSA]{Irregular Repetition Slotted ALOHA}
    \acro{SIC}[SIC]{Successive Interference Cancellation}
    \acro{CSA}[CSA]{Coded Slotted ALOHA}
    \acro{mMTC}[mMTC]{Massive Machine-type Communications}
    \acro{CSMA/CA}[CSMA/CA]{Carrier-Sense Multiple Access with Collision Avoidance}
    \acro{SA}{slotted ALOHA}
    \acro{AGV}[AGV]{automated guided vehicle}
    \acro{CACC}[CACC]{Cooperative Adaptive Cruise Control}
\end{acronym}

\section{Introduction}
In the manufacturing and automotive industry, there has been a significant increase in applications that require continuous control decisions based on the perception of a remote target. This shift has introduced new demands on wireless communication systems, prioritizing not only the reliable transmission of data points, but the delivery of the right information at the right time. Semantic communication offers a conceptual framework to quantify this goal, recognizing that conventional metrics such as latency or throughput are insufficient for evaluating such advanced communication needs. New semantic metrics are thus essential \cite{Uysal_Semantic_2022,kountouris2021semantics}.

A common semantic metric to quantify the freshness of information at the receiver is the \ac{AoI}. The \ac{AoI} quantifies the time elapsed since the most recent information available at the receiver was generated at the transmitter \cite{kaul_minimizing_2011,Yates_Age_2021}. 
Optimizing the communication system to minimize the average AoI will thus lead to fresh information at the receiver, providing a basis for continuous control decisions. For some applications such as, e.g., platooning in autonomous driving or \ac{AGV} maneuvering in automated factories, particularly strict requirements on the freshness of information are necessary to ensure the intended functionality. These requirements can be expressed using an age threshold \cite{Vinel_thresholdPAoI_2015,fiems_age--information_2024,abdel-aziz_optimized_2020,thunberg_unreliable_2021}. 
The communication system should then be optimized to minimize the age-violation probability, i.e., the proportion of time spent above the age threshold. To facilitate this form of optimization, the generation time and the age threshold need to be available on all communication layers, for example through the use of a semantic header \cite{Utkovski_Semantic_2023} or, more generally, through the introduction of a semantic-effectiveness (SE)-plane \cite{popovski_semantic-effectiveness_2020}.

\subsection{Prior Work}
The concept of information freshness has been first introduced in the context of vehicular communications \cite{kaul_minimizing_2011}, and since then the AoI has been extensively studied in the literature (see, e.g., \cite{Yates_Age_2019} for an early overview). While earlier works focused mostly on point-to-point settings, more recent approaches target multiple access settings, with the objective to re-evaluate channel access mechanisms (in particular grant-free), with regard to the AoI behavior in the system.   

For some simplified system settings, the effects of basic parameters such as, e.g., channel utilization, on the AoI are studied in \cite{Yates_Status_2017,Talak_Distributed_2018, Orhan_Analysis_2021, Baiocchi_Age_2022,Baiocchi_Model_2017,Munari_Modern_2021, Bezmenov_IRSA_2024, Cao_optimize_2022, Rolich_Impact_2023}. Among the methods already investigated are classical protocols such as \ac{SA} \cite{Yates_Status_2017,Talak_Distributed_2018} and AoI-optimized versions of \ac{SA} \cite{Orhan_Analysis_2021} or \ac{CSMA/CA} \cite{Baiocchi_Age_2022,Baiocchi_Model_2017}. In addition, modern channel access methods \cite{Book_Modern_2016}, such as \ac{IRSA} \cite{Liva_Graph_2011}, have been investigated with respect to AoI \cite{Munari_Modern_2021,Bezmenov_IRSA_2024}. Especially for set-ups such as \ac{mMTC} \cite{Munari_Modern_2021} or \ac{V2X} \cite{Bezmenov_IRSA_2024}, these protocols achieve particularly low average \ac{AoI} and age-violation probabilities. 

Another important protocol is \acf{SPS}. SPS, which falls in the class of reservation-based protocols, forms the basis of the \ac{3GPP} \ac{NR-V2X} standard \cite{3gpp_TS38321} for uncoordinated transmission on Sidelink. When optimizing an already deployed protocol such as \ac{SPS}, the specifics of the protocol should be well understood so that any adaptations remain compatible with the protocol. A detailed description of the \ac{SPS} protocol, especially with regard to interactions with other communication layers, can be found in \cite{Javier_Tutorial_5G_2021}. 
An evaluation of the protocol with regard to classical metrics such as packet error rate and throughput (e.g. as a function of the channel utilization and the distance between two vehicles) is performed in \cite{Bezmenov_Semi-Persistent_2021,gonzalez-martin_analytical_2019}.  The challenges of SPS with non-periodic sampling are investigated in \cite{Bezmenov_Semi-Persistent_2021,molina-masegosa_empirical_2020}. 

In the context of SPS, several efforts were also made to characterize the average AoI or \ac{PAoI} (defined as the value that the AoI reaches immediately before a reset triggered by an update). In \cite{Cao_optimize_2022}, a model was developed that shows the dependence between the expected \ac{PAoI} and the packet generation interval. The model considers collisions with up to two nodes. In \cite{Rolich_Impact_2023}, a model was developed that characterizes the influence of the reservation duration on the AoI and the \ac{PAoI}. 

\subsection{Contribution and Outline}
\label{subsec:contribution}
Despite extensive research in this area, the AoI-related behavior of reservation-based access protocols is not yet fully understood. In particular, prior works focus on average AoI and there does not yet exist, to the best of our knowledge, a stochastic model for the age-violation probability of the SPS protocol. The age-violation probability is the probability that the AoI exceeds, at any point in time, a predetermined maximum threshold. The age-violation probability is a particularly important quantity when optimizing a wider range of automotive and other applications \cite{Vinel_thresholdPAoI_2015,fiems_age--information_2024,abdel-aziz_optimized_2020,thunberg_unreliable_2021}.

Motivated by this, we propose a novel model for the age-violation probability of the SPS protocol. 
Initial results with regard to this topic were discussed in \cite{Bezmenov_Characterization_AoI_2022}.
Specifically, we consider a setting in which devices send time-stamped status updates to all devices in communication range on a shared channel without feedback. In this setup, we make the following key contributions: 
\begin{itemize}
    \item We propose a Markov model to characterize the AoI process, prove that it has a stationary distribution, and derive a closed-form approximation for the distribution of AoI in the stationary state.
    \item Through empirical analysis and in-depth discussion, we identify conditions under which the assumptions used in our derived model, as well as those presented in \cite{Bezmenov_Characterization_AoI_2022} and \cite{Rolich_Impact_2023}, are applicable.
    \item Utilizing this stationary distribution, we analyze both the average AoI and the age-violation probability (i.e., the probability that the AoI exceeds a certain pre-defined threshold). We examine the critical influence of the reservation duration parameter and the sampling rate of perceptual data. Our findings show that in some important cases, optimizing for average AoI does not lead to an optimal age-violation probability. This shows how important it is to optimize directly for age violation probability in applications where robustness and reliability in edge cases are more important than the average performance of the system. 
\end{itemize}
\noindent\textit{Paper Outline:}
The remainder of the paper is organized as follows.  In Section~\ref{sec:System_Model}, we introduce the system model, focusing on the key features of the studied reservation-based protocol. In Section~\ref{sec:mainResults}, we present a closed-form solution for the probability distribution of the AoI and discuss numerical results for the average AoI and the age-violation probability and compare those to results derived by simulation. 
The assumptions necessary to derive the probability distribution of the AoI, along with a discussion of the specific conditions under which these assumptions are applicable, are presented in Section~\ref{sec:assumptions}. Section~\ref{sec:conclusion} summarizes the results of the paper.\\

\noindent\textit{Notation:} We use $\mathbb{N}=\{1,2,\dots\}$ to denote the set of natural numbers and $\mathbb{N}_0$ to denote $\mathbb{N}\cup \{0\}$. The indicator function is denoted as $\mathds{1}_A$ in the domain $A$. Further, we use
\begin{equation}
    \label{def:Binomial}
        \binomial{n}{p}{k}:=\binom{n}{k}p^k(1-p)^{n-k}
\end{equation}
to denote the \ac{pmf} of the binomial distribution with $n$ trials and success probability $p$.
We define
\begin{equation}
    \label{def:bigAst}
    \bigast_{i=0}^{j} \big( f_{i}(\cdot)\big)(x) := \big(f_{0} \ast f_{1} \ast \dots \ast f_{j}\big)(x)
\end{equation}
to be successive convolutions.
\section{System Model}
\label{sec:System_Model}
{We consider a network consisting of $\nodesSystem \in \natOne$ homogeneous nodes, as illustrated in Fig.~\ref{fig:all_to_all}. Each node monitors a process (e.g., the velocity of a vehicle or the temperature measured at a sensor), which is of interest to all neighboring nodes. Consequently, every node must communicate its observations to all other nodes, leading to an all-to-all communication. 
For that purpose the nodes share a communication channel, which we model as slotted in time (i.e., all time-related quantities are expressed as multiples of the duration of one slot). Assuming a \textit{frame size} $\frameSize \in \natOne$ ($\frameSize>\nodesSystem$) and that all nodes are frame and slot synchronous, we decompose the time $\variableT \in \natZero$ into a \textit{frame index} $\frameIndex{\variableT} \in \natZero$ and \textit{frame position} $\position{\variableT}\in \{0,\dots,\frameSize-1\}$ as
\begin{align}
    \frameIndex{\variableT}&:=\lfloor \variableT/\frameSize\rfloor\\
    \position{\variableT}&:=t-\frameIndex{\variableT}\frameSize.
\end{align}

\begin{figure}
    \centering
    \includegraphics[width=0.95\linewidth]{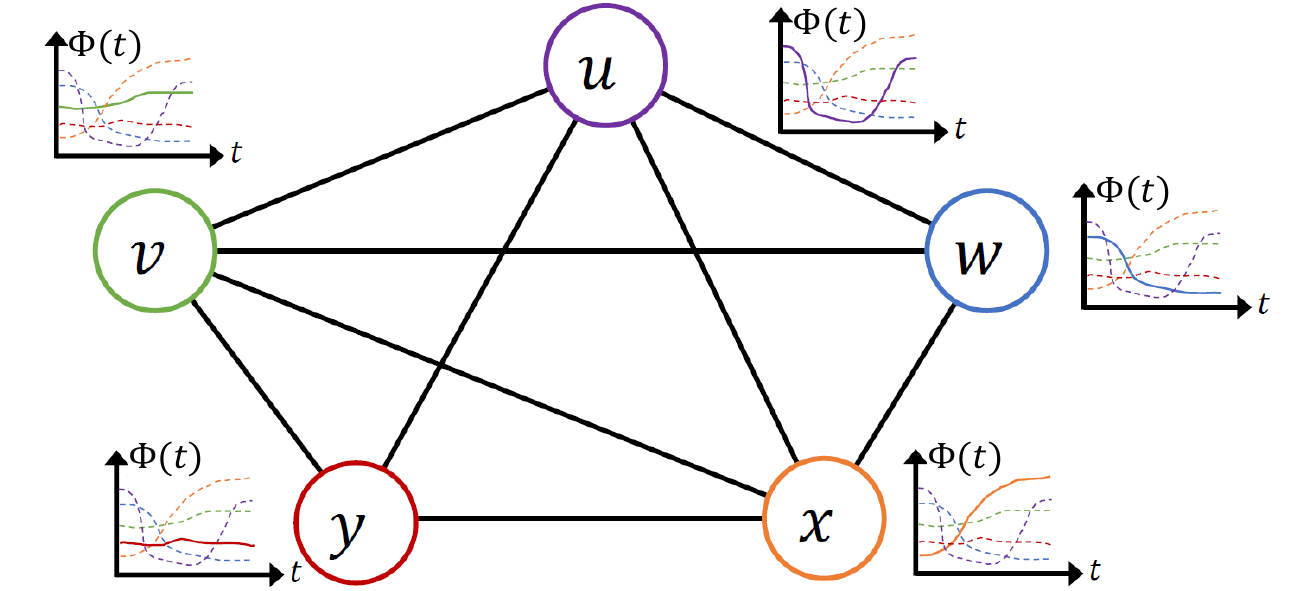}
    \caption{\textbf{All-to-all communication.} Each node monitors its own time process, which is of interest to all neighboring nodes. Consequently, every node must communicate its observations to all other nodes, leading to an all-to-all communication.}
    \label{fig:all_to_all}
\end{figure}

\begin{figure}
    \centering
    \includegraphics[width=0.95\linewidth]{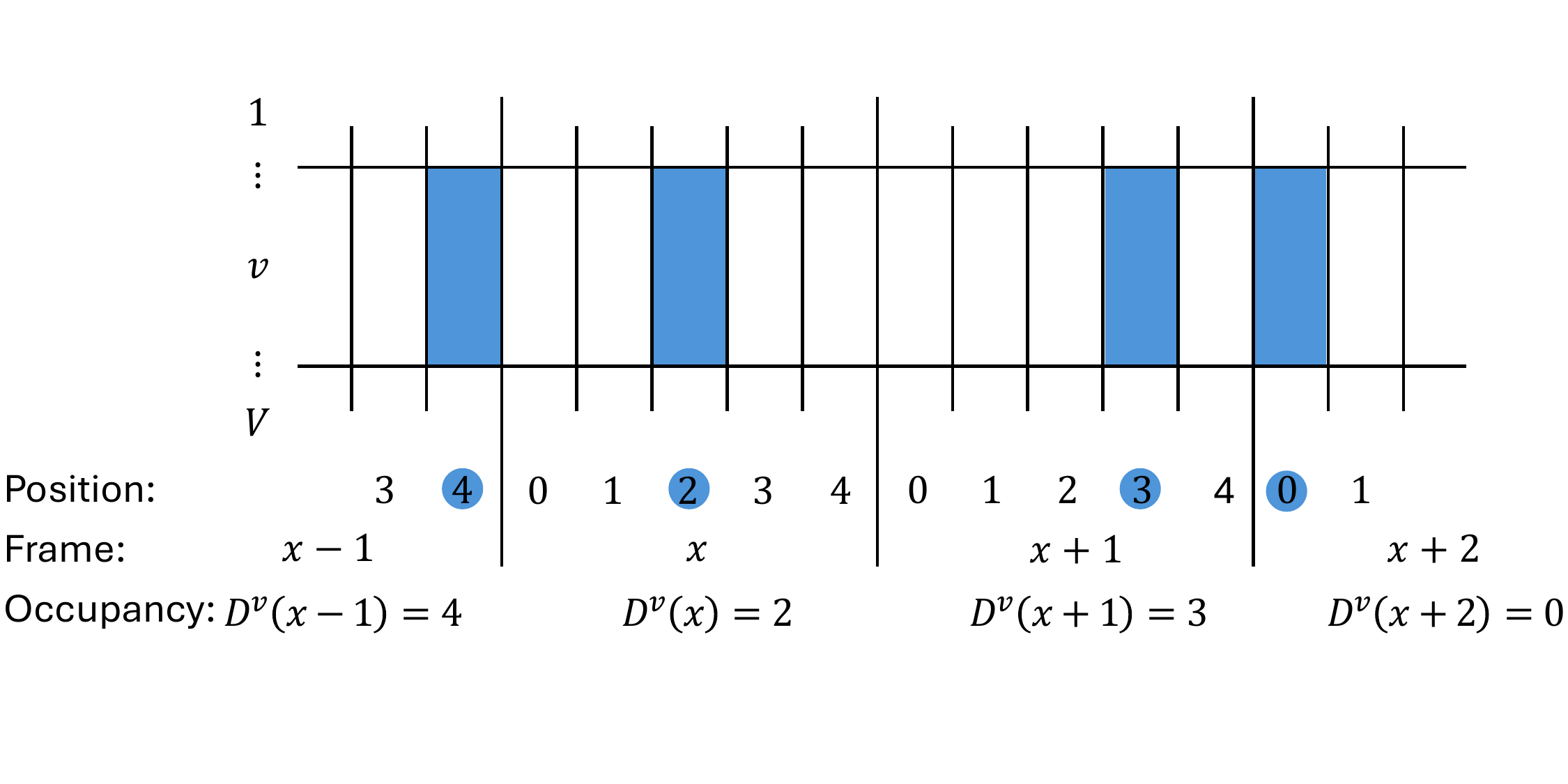}
    \caption{\textbf{Channel occupancy pattern of node $\nodeIndexF$ by frame.} A time-slotted communication channel is split into frames. Each frame consists of five slots, which are numbered and called positions. The position in which node $\nodeIndexF$ transmits in each frame $\variableX$ is marked blue and described by $\selectedSlot{\nodeIndexF}{\variableX}$.}
\label{fig:frame_structure}
\end{figure}

In the analysis we consider an arbitrary but fixed node\footnote{This assumption will simplify our analysis, but we can make it without any loss of generality since the same analysis will also hold for any other choice of $\nodeIndexF$.} $\nodeIndexF \in \{1,\dots,\nodesSystem\}$. This node $\nodeIndexF$ monitors a time-varying process $\process(\variableT)$ which represents some information that is needed by applications running on the other nodes. This process $\process(\variableT)$ is sampled at the beginning of each frame (i.e., at position zero), corresponding to periodic sampling with periodicity $\frameSize$. The content of a sample can be transmitted in exactly one slot. The position in which node $\nodeIndexF$ transmits during frame $\variableX$ is denoted as $\selectedSlot{\nodeIndexF}{\variableX} \in \{0,\dots,\frameSize-1\}$. Thus, the channel occupancy pattern (see Fig.~\ref{fig:frame_structure}) can be described as a sequence $\selectedSlotVector{\variableX} = \left({\selectedSlot{1}{\variableX},\dots, \selectedSlot{\nodesSystem}{\variableX}}\right)_{\variableX \in \natZero}$. In the following, we will refer to $\selectedSlot{\nodeIndexF}{\variableX}$ as $\selectedSlot{}{\variableX}$.

The number of simultaneous transmissions with node $\nodeIndexF$ on the shared channel in frame $\variableX$ is denoted as the \textit{frame state} $\frameState{\variableX}$ and defined as
\begin{equation}
\label{def:frameState}
  \frameState{\variableX} := \sum_{i=1}^{\nodesSystem}\mathds{1}_{\selectedSlot{i}{\variableX}=\selectedSlot{}{\variableX}}\text{.}
\end{equation}

In this system model, we consider nodes that rely on half-duplex for sending and receiving. This means that a node cannot send and receive at the same time. Further, all nodes are in communication range of each other and, as a consequence, there are no hidden nodes in the system. We discuss the implications of this assumption in Section~\ref{subsec:full_connectivity}. Considering a collision channel model, we interpret
\begin{itemize}
    \item $\frameState{\variableX} =1$ as a \emph{singleton}, i.e., the transmission of node $\nodeIndexF$ was not simultaneous with any other transmission. In this case, we assume successful decoding at every node in the network.
    \item $\frameState{\variableX} \geq 2$ as a \emph{collision}, i.e., the transmission of node $\nodeIndexF$ was simultaneous with at least one other transmission. In this case, we assume that neither transmission can be decoded correctly.
\end{itemize}

\subsection{Age of Information}
For a node $\nodeIndexS \neq \nodeIndexF \in \{1,\dots, \nodesSystem\} $, the \ac{AoI} describes the freshness of node $\nodeIndexS$'s maintained information about $\process(\variableT)$ at time $\variableT$. Formally, the \ac{AoI} $\Delta^{\nodeIndexF \rightarrow\nodeIndexS}(\variableT)$ is defined as the time between the sampling of the last singleton transmission and the time of observation $\variableT $ \cite{kaul_minimizing_2011}. Since collisions and half-duplex operation are the only error sources we model, all nodes receive the same information at the same time. Thus, without loss of generality, we can drop $\nodeIndexF \rightarrow\nodeIndexS$ and use $\AoI{\variableT}$ to denote the \ac{AoI}. In the considered system model, node $\nodeIndexF$ samples $\process$ at the beginning of each frame at position 0. Thus, we define the \ac{AoI} as 
\begin{align}
\label{def:AoI}
    \AoI{\variableT} & :=\begin{cases}
    \frameSize \cdot \collidedFrames{\frameIndex{\variableT} -1 } + \frameSize + \position{\variableT} & \position{\variableT} < \selectedSlot{}{\frameIndex{\variableT}} \\
    \frameSize \cdot \collidedFrames{\frameIndex{\variableT}} +  \position{\variableT} & \position{\variableT} \geq \selectedSlot{}{\frameIndex{\variableT}}\text{,}
    \end{cases}
\end{align} 
where the \textit{collision duration} $\collidedFrames{\variableX}$ describes the number of frames since the most recent update of node $\nodeIndexF$ in frame $\variableX$, and is defined as

\begin{equation}
\label{def:collidedFrames}
    \collidedFrames{\variableX} := \min\{\rCollidedFrames \in \{0,\dots,\variableX\}:~ \frameState{\variableX-\rCollidedFrames} = 1\}\text{.}
\end{equation}

In order to exclude the possibility of this and similar minima definitions being over an empty set, we use the convention that $\selectedSlotVector{0}=(0,\dots,0),\selectedSlotVector{1}=(1,\dots,\nodesSystem)$. We will see in Lemma~\ref{lem:stationarity} that the system can be described by a Markov chain with a unique stationary distribution. Since our later analysis is confined to the stationary state of the system, this convention does not affect the results of our analysis.

\begin{figure}
    \centering
    \includegraphics[width=0.7\linewidth]{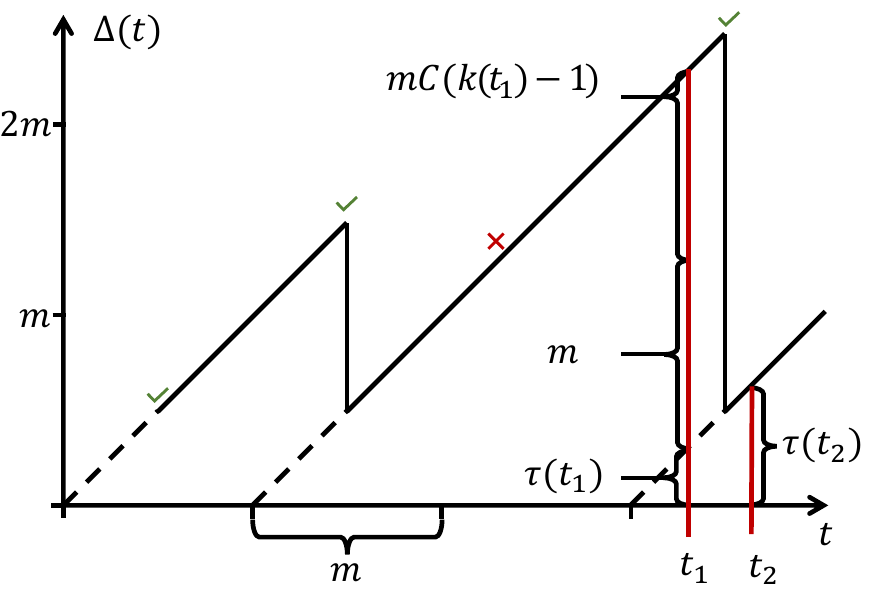}
    \caption{\textbf{Example of realization of the AoI $\AoI{\variableT}$.} The time is structured into frames of $\frameSize$ slots. At the beginning of each frame, a packet is generated (sampled from $\process(\variableT)$) and afterwards transmitted. The green checks mark signify singleton transmissions while the red crosses mark collided transmissions. The solid line shows the AoI at the receiver. In addition, the dashed line shows the channel access delay. For this example with a singleton transmission in frame $\frameIndex{\variableT_1}$, respectively $\frameIndex{\variableT_2}$, and a collided transmission in the previous frame, the collision duration $\collidedFrames{\frameIndex{\variableT_2}}$ equals zero according to \eqref{def:collidedFrames}.}
    \label{fig:AoI_def}
\end{figure}
In Fig.~\ref{fig:AoI_def}, we show an exemplary evolution of the AoI.
\subsection{Channel Access Model}
\label{subsec:transmissionRules}
In this paper, we consider decentralized access to a shared channel based on the 3GPP standard (see, e.g., 
 \cite[Section 5.22.1.1]{3gpp_TS38321}). 
The SPS protocol is designed to avoid collisions by implementing a resource reservation mechanism. All nodes reserve a position within the frame for future transmissions. By listening to the channel, other nodes can anticipate transmissions and avoid transmitting in the same position avoiding collisions. In Fig.~\ref{fig:protocol_SPS}, we illustrate an example of the protocol operation showing the transmission pattern for a system with four nodes in two subsequent frames.
 

According to the main idea of SPS, new reservations are only made on positions that were empty in the previous frame. The empty positions in frame $\variableX$ are defined as
\begin{equation}
\label{def:emptySlots}
\emptySlots{\variableX} := \{0,\dots,\frameSize-1\} \backslash \selectedSlotVector{\variableX} \text{.}
\end{equation}
Considering the example in Fig.~\ref{fig:protocol_SPS}, this would correspond to $\emptySlots{\variableX} = \{2,3\}$ and $\emptySlots{\variableX+1} = \{0,1,2\}$. From the empty positions, the reselecting node selects one position uniformly at random. To denote the number of empty slots in frame $\variableX$, we use 
\begin{equation}
\label{def:numEmptySlots}
    \numEmptySlots{\variableX}:=|\emptySlots{\variableX}|\text{.}
\end{equation}

At the beginning of each new reservation, the \textit{reservation counter} is selected from a predetermined distribution. 
Each time the node makes a transmission at the selected position, the reservation counter is decreased by one until it reaches zero, and a new reservation begins. 
For simplicity, we assume that each new reservation begins in a different position. Later, in Section~\ref{sec:application}, we discuss how to extend the model to cover new reservations in the same position, introducing the parameter \textit{sl-ProbResourceKeep} as defined in \cite{3gpp_TS38321}.

Considering a reservation counter drawn from a geometric distribution\footnote{The primary focus of our analysis is to understand the impact of reservations on the system's dynamics. For a clearer and more structured analysis, we use a geometric distribution, even though in practice a finite distribution, such as the uniform distribution, is commonly used. A comparison of the geometric and uniform distributions in Fig.~\ref{fig:expected_AoI} shows that the geometric distribution provides a good approximation of the uniform distribution in certain configurations.} with success probability $\pReservationEnd$ (henceforth called the \textit{ending probability}), the transmission rules can equivalently be summarized as follows:

At $\variableT = \frameSize\variableX $, before the first transmission of frame $\variableX$ is made, each node $\nodeIndexF$ 
\begin{itemize}
    \item with probability $(1-\pReservationEnd)$ keeps the same position \\ $\selection = \selectedSlot{}{\variableX-1}$ for the next transmission; 
    \item with probability $\pReservationEnd$ selects a new position \\$\selection \sim \Uniform{\emptySlots{\variableX-1}}$ for the next transmission. 
\end{itemize}
The transmission then takes place at $\variableT'= \frameSize\variableX  + \selection$. 
We refer to the selection of the new position as \textit{reselection}.
\begin{figure}
    \centering
    \subfloat[frame $\variableX$\label{fig:protocol_SPS_1}]{%
       \includegraphics[width=0.48\linewidth]{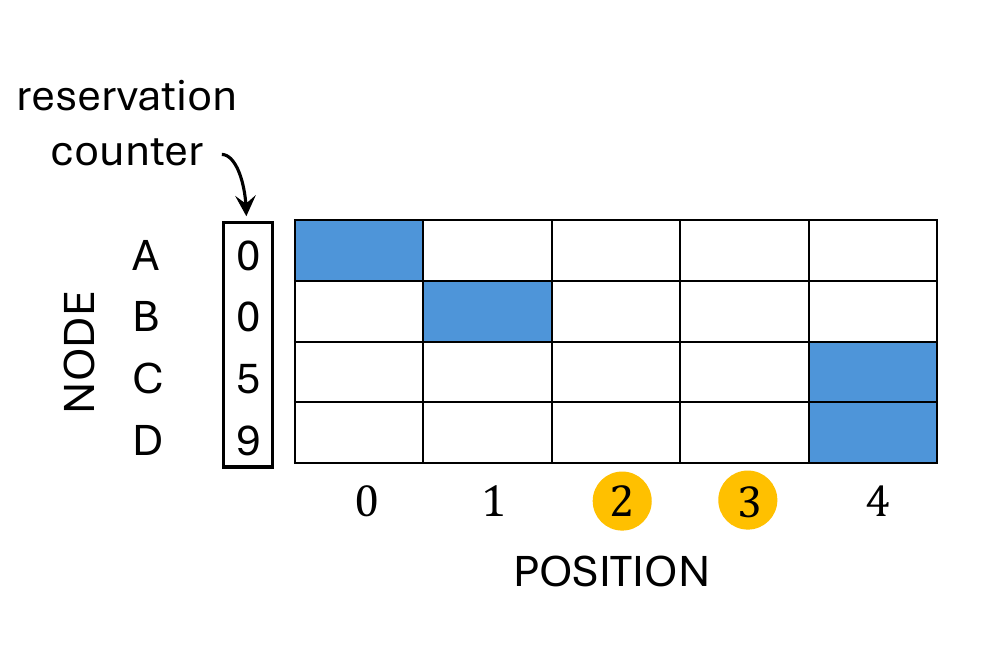}}
    \hfill
    \subfloat[frame $\variableX+1$\label{fig:protocol_SPS_2}]{%
        \includegraphics[width=0.48\linewidth]{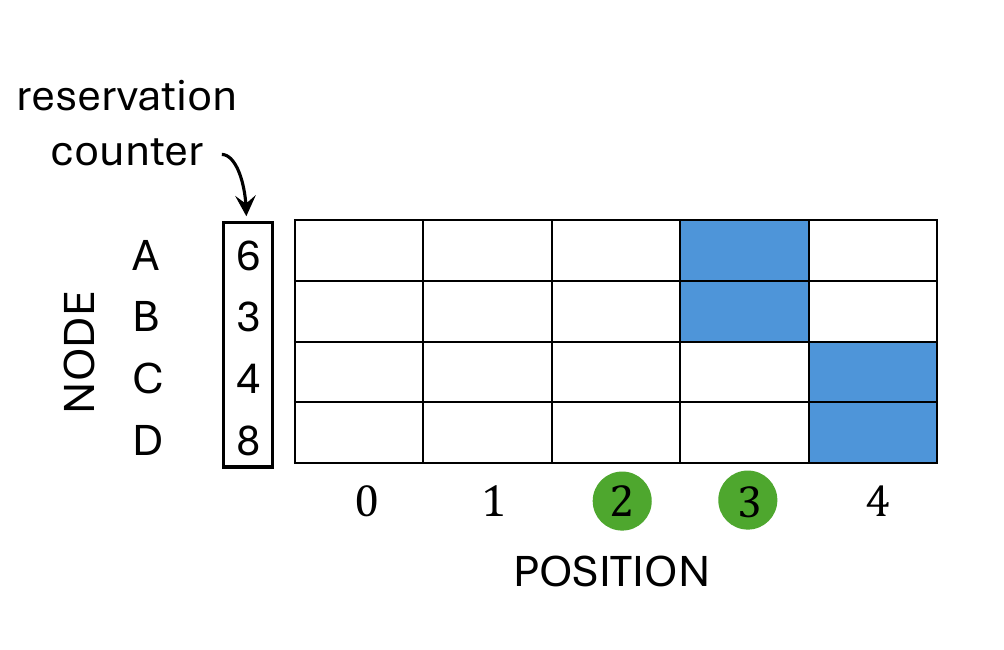}}
    \caption{\textbf{Example of operation for \ac{SPS}.} (a) shows the transmission pattern and reservation counter for nodes A, B, C and D during a generic frame $\variableX$ and (b) shows the subsequent frame $\variableX+1$. In frame $\variableX$, the reservation counters of node C and D are greater than zero, therefore they transmit in frame $\variableX + 1$ in the same position as in frame $\variableX$ and the reservation counter is decremented by one. In contrast, the reservation counters of node A and B are both zero in frame $\variableX$, therefore both nodes perform a reselection and select new positions for their transmissions in frame $\variableX+1$. The new positions are selected uniformly from the empty positions (orange) in frame $\variableX$. The number of remaining reservations is reset to a new value drawn from a uniform distribution. In the example shown, both nodes A and B select position 3 for their new transmission. Their transmissions collide, and the collision can potentially last for several frames until one of the nodes performs a reselection.}
    \label{fig:protocol_SPS}
\end{figure}


In this context we define the \textit{reservation duration} $\OngoingTransmissions{\variableX}$ to denote at frame $\variableX$ the number of frames since node $\nodeIndexF$ changed its transmission slot as
\begin{equation}
\label{def:OngoingTransmissions}
    \OngoingTransmissions{\variableX} := \min\{\rOngoingTransmissions \in \{1,\dots,\variableX\}:~ \selectedSlot{}{\variableX-\rOngoingTransmissions} \neq\selectedSlot{}{\variableX-\rOngoingTransmissions+1}\}\text{.}
\end{equation}

We use $\FrameRecursive{\indexN}{\variableX}$ to characterize the frame at which the $\indexN$-the reservation \textit{ends} prior to frame $\variableX$, where $\FrameRecursive{\indexN}{\variableX}$ is recursively defined as
\begin{equation}
    \label{def:frameRecursive}
    \FrameRecursive{\indexN+1}{\variableX}:=\FrameRecursive{\indexN}{\variableX}-\OngoingTransmissions{\FrameRecursive{\indexN}{\variableX}},
\end{equation}
with $\FrameRecursive{0}{\variableX}:=\variableX$. 
To characterize the \textit{start} of the $\indexN$-th reservation, we use $\reservationStart{\indexN}{\variableX}:=\FrameRecursive{\indexN+1}{\variableX}+1$. We use node C in Fig.~\ref{fig:protocol_SPS} as an example to explain the notation. According to the definition, in frame $\variableX+ 1$, the start and end of the zeroth reservation would be in frame $\reservationStart{0}{\variableX+1}=\FrameRecursive{0}{\variableX+1}=\variableX+1$. In addition, the end of the first reservation would be in frame $\FrameRecursive{1}{\variableX+1}=\variableX$, while the start of the first reservation is not shown in the example.

\section{Main Results}
\label{sec:mainResults}
In this chapter, we present two results. The first result, Theorem~\ref{thm:segmentation}, describes the segments into which the AoI can be decomposed. The second and main result of this paper, Approximation ~\ref{thm:pmfAoI}, is an approximation of the pmf of the AoI. In addition, we show numerical results for the cdf of the approximated AoI and discuss the quality of the approximation by comparing it to simulations of the above discussed system model. Finally, we conclude this section with numerical results for the average AoI and the age-violation probability.
\subsection{Decomposing the AoI}
The AoI $\AoI{\variableT}$ can be expressed as the sum of different segments. To separate the segments, we use the \textit{collided reservation count} $\collidedReselections{\variableX}$ to describe the number of reselections node $\nodeIndexF$ has performed up to frame $\variableX$ since its last singleton transmission as
\begin{equation}
    \label{def:collidedReselections}
    \collidedReselections{\variableX} := \min\{\indexN \in \{0,\dots,\variableX\}: \frameState{\FrameRecursive{\indexN}{\variableX}}=1\}.
\end{equation}

\begin{theorem}
\label{thm:segmentation}
The AoI for node $\nodeIndexF$ at time $\variableT$ can be written as
\begin{equation}
    \label{thm:segmantationtD}
    \AoI{\variableT}
    =
    \position{\variableT}  + \frameSize \bigg( 1+ \sum\nolimits_{\rCollidedReselections=0}^{\collidedReselections{\frameIndex{\variableT} -1 }-1}\OngoingTransmissions{\FrameRecursive{\rCollidedReselections}{\frameIndex{\variableT} -1}}\bigg)
\end{equation}
if $\position{\variableT} < \selectedSlot{}{\frameIndex{\variableT}}$ and as
\begin{equation}
    \label{thm:segmantationDt}
    \AoI{\variableT}
    =
    \position{\variableT} + \frameSize \sum\nolimits_{\rCollidedReselections=0}^{\collidedReselections{\frameIndex{\variableT} }-1}\OngoingTransmissions{\FrameRecursive{\rCollidedReselections}{\frameIndex{\variableT}}}
\end{equation}
if $\position{\variableT} \geq \selectedSlot{}{\frameIndex{\variableT}}$.
\end{theorem}
For a better understanding of the result, the theorem is illustrated with an example in Fig.~\ref{fig:thm_seperation}. The proof of Theorem~\ref{thm:segmentation} can be found in Appendix~\ref{sec:proofsSegmentation}.

\begin{figure*}
    \centering
    \subfloat[Theorem~\ref{thm:segmentation} -- equation \eqref{thm:segmantationtD}: The transmission in frame $\frameIndex{\variableT}$ has occurred before $\variableT$ \label{fig:thm_sepAoI_Dt} ]{%
       \includegraphics[width=0.7\linewidth]{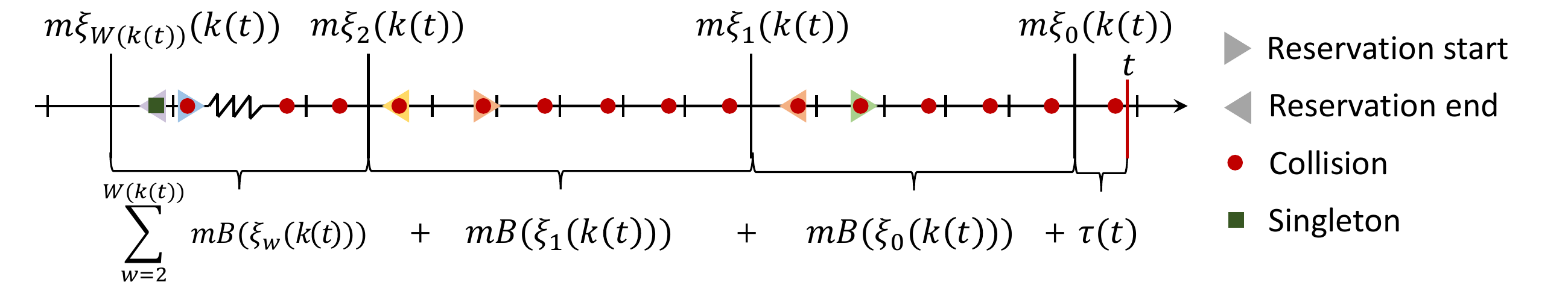}}\\[+2ex]
    \subfloat[Theorem~\ref{thm:segmentation} -- equation \eqref{thm:segmantationDt}: The transmission in frame $\frameIndex{\variableT}$ has occurred after $\variableT$  \label{fig:thm_sepAoI_tD}]{%
        \includegraphics[width=0.7\linewidth]{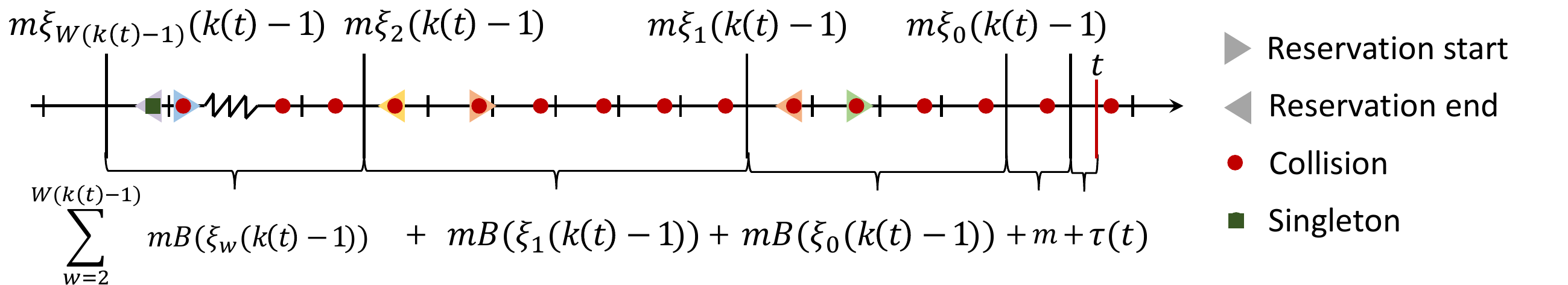}}        
    \caption{\textbf{Segments contributing to the AoI.} The transmission pattern of node $\nodeIndexF$ is illustrated frame by frame. The transmissions are shown here as circles. Green squares signify singleton transmissions and red circles collided transmissions. In addition, we mark the start and end of a reservation using inward-pointing triangles. Based on \eqref{def:frameRecursive} and \eqref{def:collidedReselections}, we label the final transmission of the last reservations before $\frameIndex{\variableT}$ (a) and, $\frameIndex{\variableT}-1$ (b), respectively. Thus, the distance between two reservations, for example, the first and second reservation before frame $\frameIndex{\variableT}$, is $\frameSize\OngoingTransmissions{\FrameRecursive{1}{\frameIndex{\variableT}}}$. According to definition~\eqref{def:AoI}, the AoI is the time between the sampling instance of the last singleton and the time of observation $\variableT$. In this example the last singleton is in frame $\FrameRecursive{\collidedReselections{\frameIndex{\variableT}}}{\frameIndex{\variableT}}$(a) respectively $\FrameRecursive{\collidedReselections{\frameIndex{\variableT}-1}}{\frameIndex{\variableT}-1}$(b) and the sampling occurred at the beginning of this frame. The last singleton transmission before $\variableT$ is always the last transmission within a reservation. Thus, the AoI is characterized by the sum of $\rCollidedReselections$ reservation durations (in slots) plus the slots already passed in frame $\frameIndex{\variableT}$.} 
    \label{fig:thm_seperation}
\end{figure*}
\subsection{Approximation results for the AoI}
\label{subsec:pmfAoI}
The core result of this paper is an approximation that describes the pmf of the AoI in stationary state (see Section~\ref{subsec:stationarity} especially Lemma~\ref{lem:stationarity}). The approximation is based on the assumption discussed in Section~\ref{subsec:AssIndependence} implying that the dependence of the state at the start of a new reservation on the state of the previous reservation is negligible. A more detailed discussion can be found in Section~\ref{subsec:AssIndependence}. 
\begin{approximation}
\label{thm:pmfAoI}
For a stationary system, the probability distribution of the AoI $\probM{\AoI{\variableT}= \rAoI}$ can be approximated as 
    \begin{align}
    \label{eq:thmPmfAoI1}
        \probM{\AoI{\variableT} = \rAoI}  \approx 
        \frac{\position{\variableT}}{\frameSize} \functionQ{\frameIndex{\rAoI}-1} 
            + \frac{\frameSize-\position{\variableT}}{\frameSize} \functionQ{\frameIndex{\rAoI}},
    \end{align}
if $\position{\variableT}=\position{\rAoI} $ and as 
    \begin{align}
    \label{eq:thmPmfAoI2}
        \probM{\AoI{\variableT} = \rAoI}  = 0
    \end{align}
otherwise, with
    \begin{align}
    \label{prob:collidedFramesApprox}
        \functionQ{\rCollidedFrames}:= 
        \sum_{\rCollidedReselections=0}^{\maxCollidedReselections} \left(
            \sum_{\rOngoingTransmissions=1}^{\maxOngoingTransmissions} 
            \bigg( \pReservationEnd(1-\pReservationEnd)^\rOngoingTransmissions 
         - \functionP{\rOngoingTransmissions} \bigg)\right)\cdot
            \bigast_{\indexN=0}^{\rCollidedReselections-1}
            \bigg(\functionP{\cdot}\bigg)(\rCollidedFrames) 
    \end{align}
and 
    \begin{align}
        \functionP{\variableX} &:= \pReservationEnd(1-\pReservationEnd)^\variableX 
        \cdot \Bigg(1- \sum_{\rFrameState=1}^{\nodesSystem-1} \bigg[\left(1-\left(1-\pReservationEnd\right)^{\variableX-1}\right)^{\rFrameState}\nonumber\\
        &\cdot
        \sum_{\rNumEmptySlots=1}^{\frameSize} 
            \probNumEmptySlots{\rNumEmptySlots} \cdot
            \binomial{\nodesSystem-1}{\frac{\pReservationEnd}{\rNumEmptySlots}}{\rFrameState}
        \bigg]\Bigg),
    \end{align} 
\end{approximation}
for $\variableX \in\{1,\maxOngoingTransmissions\}$ and with 
\begin{equation}
\label{def:emptySlotsStatinarity}
    \probNumEmptySlots{\rNumEmptySlots}:= \probM{\numEmptySlots{\variableX}= \rNumEmptySlots}
\end{equation}
denoting the stationary probability distribution of empty slots $\numEmptySlots{\variableX}$ in frame $\variableX$. A detailed explanation on the existence of a stationary distribution $\probNumEmptySlots{\rNumEmptySlots}$ can be found in Section~\ref{subsec:stationarity}.
The parameter $\maxCollidedReselections$ describes the maximum number of collided reservations that we take into account within the model. The parameter $\maxOngoingTransmissions$ describes the maximum number of transmissions per reservation that we take into account within the model. The accuracy of the approximation increases as $\maxCollidedReselections$ and $\maxOngoingTransmissions$ increase. A detailed derivation of Approximation~\ref{thm:pmfAoI} can be found in Appendix~\ref{sec:derivation_pmf}.
Clearly, the pmf approximation of the AoI given in Approximation~\ref{thm:pmfAoI} depends on the time of observation $\variableT$. However, it is worth noting that this dependence concerns only the position $\position{\variableT}$ and not the frame index $\frameIndex{\variableT}$. For a compact representation of the different distributions for the different positions, we average the distribution of the AoI over all possible positions and represent them as 
\begin{equation}
    \probM{\aAoI=\rAoI} := \sum_{n=0}^{m-1} \left( \frac{1}{\frameSize}\probM{\AoI{\frameIndex{t}\frameSize+\indexN}=\rAoI}\right).
\end{equation}
\subsection{Numerical results}
\label{subsec:numericalResults}
In order to validate the accuracy of the pmf we have implemented the system described above in Python for comparison. We show in particular that the approximation reproduces the pmf well despite the assumption of independent states. The source code used for the simulations is available as an electronic supplement with this paper. Each individual simulation consists of 550,000 frames, whereby the first 50,000 frames are not evaluated in order to avoid transient effects. 
For the numerical evaluation of Approximation~\ref{thm:pmfAoI}, we consider reservation durations up to a length of $\maxOngoingTransmissions=1000$ and a maximum collision number of $\maxCollidedReselections=50$. While higher values of $\maxCollidedReselections$ and $\maxOngoingTransmissions$ lead to a more precise approximation, they come at the trade-off of longer computational time. Testing different combinations showed sufficiently good results for $\maxOngoingTransmissions=1000$ and $\maxCollidedReselections=50$, as depicted in Fig.~\ref{fig:pmf_AoI_V195_m200}.
As describing $\probNumEmptySlots{\rNumEmptySlots}$ is difficult, we assume, as in \cite{Rolich_Impact_2023}, that the number of empty slots in each frame is approximately equal to the expected value of empty slots $\eNumEmptySlots$. We discuss this assumption and the derivation of $\eNumEmptySlots$ in more detail in Section~\ref{subSec:emptySlots}.

Based on this assumption, $\functionP{\rOngoingTransmissions}$ can be approximated as
    \begin{align}
        \functionP{\rOngoingTransmissions} &\approx \pReservationEnd(1-\pReservationEnd)^\rOngoingTransmissions 
        \cdot \Bigg(1- \sum_{\rFrameState=1}^{\nodesSystem-1} \bigg[\left(1-\left(1-\pReservationEnd\right)^{\rOngoingTransmissions-1}\right)^{\rFrameState}\nonumber\\
        &\cdot
            \binomial{\nodesSystem-1}{\frac{\pReservationEnd}{\eNumEmptySlots}}{\rFrameState}
        \bigg]\Bigg).
    \end{align}

In Fig.~\ref{fig:pmf_AoI_V195_m200}, we show $\probM{\aAoI\leq\rAoI}$ for different combinations of frame size (sampling period), channel load, and the ending probability $\pReservationEnd$. For all combinations of parameters shown, the proposed approximation matches the simulation results closely.  

Let us now take a closer look at the evolution of the cdf of the AoI $\probM{\aAoI\leq\rAoI}$. For all combinations, it is noticeable that the distribution has a strong growth before $2\frameSize$ and a slow growth with a long tail afterwards. The higher the AoI, the longer there has been no successful transmission. As in \cite{Rolich_Impact_2023}, we refer to time periods without successful singleton transmissions as \textit{OFF-times}. The length of the tail represents the duration of OFF-times, and the slope represents the rate at which OFF-times are observed. 

Furthermore, we observe that smaller $\pReservationEnd$ and therefore longer reservation counters (compare green and blue curves) lead to significantly higher probabilities below $2\frameSize$ and a longer tail. This means that there are fewer prolonged OFF-times, but those that remain can be even longer. When we increase the channel load $\frac{\nodesSystem}{\frameSize}$ (compare the orange and blue curve), we observe that the number of OFF-times is becoming more frequent and longer. Reducing the slots $\frameSize$ per frame by half (compare the red and blue curves) leads to a shift; the values adopted by the AoI are only half as large.

\begin{figure}
    \centering
    \includegraphics[width=0.95\linewidth]{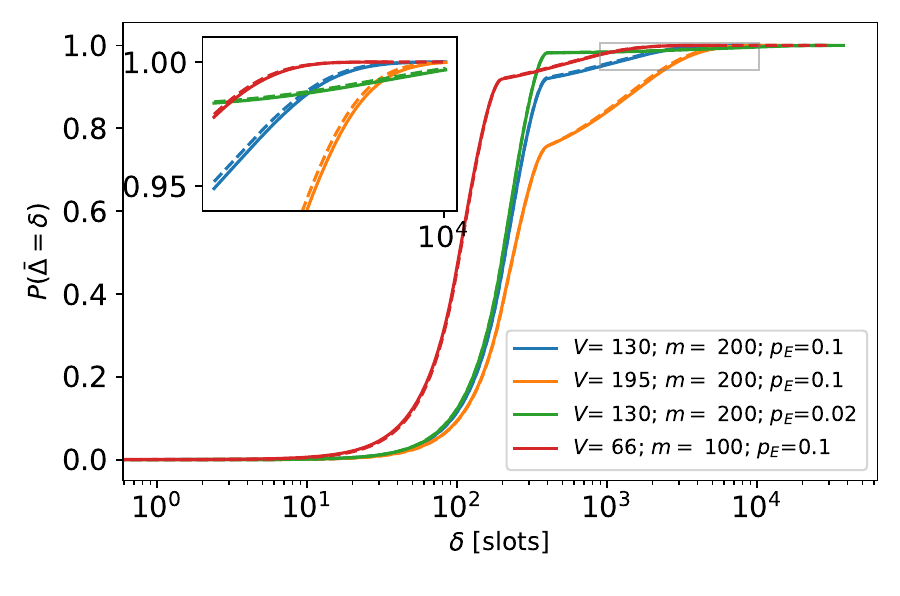}
    \caption{Cdf of the AoI comparing analytic (solid line) and simulative (dashed line) results.} 
    \label{fig:pmf_AoI_V195_m200}
\end{figure}

\subsection{Application of the AoI distribution}
\label{sec:application}
These correlations can now be used to adapt the parameters to the requirements of the respective application.
Consider, for example, the role of the reservation counter. In the \ac{3GPP} standard, the distribution from which the reservation counter is chosen is given as a function of the sampling rate (e.g. uniform in $[5,15]$ for a sampling rate of 100ms). 
In our model, we assume a geometric distribution of the reservation counter. We use the pmf from Approximation~\ref{thm:pmfAoI} to calculate the average AoI and the age-violation probability and thus determine the optimal working conditions. 

The average AoI can be computed from the pmf as 
\begin{equation}
    \eAoI = \sum_{\delta=0}^{\infty}\delta\cdot\probM{\aAoI=\rAoI} 
\end{equation}
and the age-violation probability as
\begin{equation}
    \violation{\threshold} = 1-\sum_{\delta=0}^{\threshold}\probM{\aAoI=\rAoI}\text{.}
\end{equation}

\subsubsection*{Average AoI} For the average AoI of SPS, there is already a model \cite{Rolich_Impact_2023} against which we compare our results.
The assumptions of the system model in \cite{Rolich_Impact_2023} are largely identical to ours. $\nodesSystem$ nodes send samples based on the rules presented in Section~\ref{sec:System_Model}.
The first main difference is the introduction of probability $\pKeep$, which specifies the probability with which the same slot $\selectedSlot{}{\variableX}$ is used for the duration of a further reservation. If the reservation counter is drawn from a geometric distribution with probability $\pRC$, the parameter $\pReservationEnd$ in this paper is related to the parameters $\pRC$ and $\pKeep$ in \cite{Rolich_Impact_2023} as follows:
\begin{equation}
\label{eq:pReservationEnd}
    \pReservationEnd = (1-\pRC)(1-\pKeep)\text{.}
\end{equation}
The second difference is the sample generation. In both \cite{Rolich_Impact_2023} and our model, a new sample is assumed for each frame. In our model, we additionally consider the time of sampling. We assume that samples are always generated at a fixed time in the frame (e.g. in the 0-th slot) and thus model a periodic upper-layer application. 

Fig.~\ref{fig:expected_AoI} shows the average AoI as a function of $\pReservationEnd$ for 195 nodes and 200 slots per frame. As a baseline, we used the analytic results of \cite{Rolich_Impact_2023} and simulation results. 
We include simulation results for two cases. In the first case, we assume that the reservation counter is drawn uniformly at random from $[5,15]$ according to \cite{3gpp_TS36213}. In the second case, we assume a geometric distribution with $\pRC=0.9$. In both cases $\pKeep$ is varied from 0 to 0.8, resulting, according to \eqref{eq:pReservationEnd}, in $\pReservationEnd$ varying from 0.02 to 0.1

First, we observe that the curves of our model are very similar to those of \cite{Rolich_Impact_2023}. The difference between the curves is caused by the slightly different system models, specifically our additional consideration of the sampling instance. This shows that both models can reproduce the average AoI well. As in \cite{Rolich_Impact_2023}, we also see here that in the range of small values $\pReservationEnd$, a uniform distribution of the duration of the reservation can reasonably be approximated by a geometric distribution.
\begin{figure}
    \centering
    \includegraphics[width=0.9\linewidth]{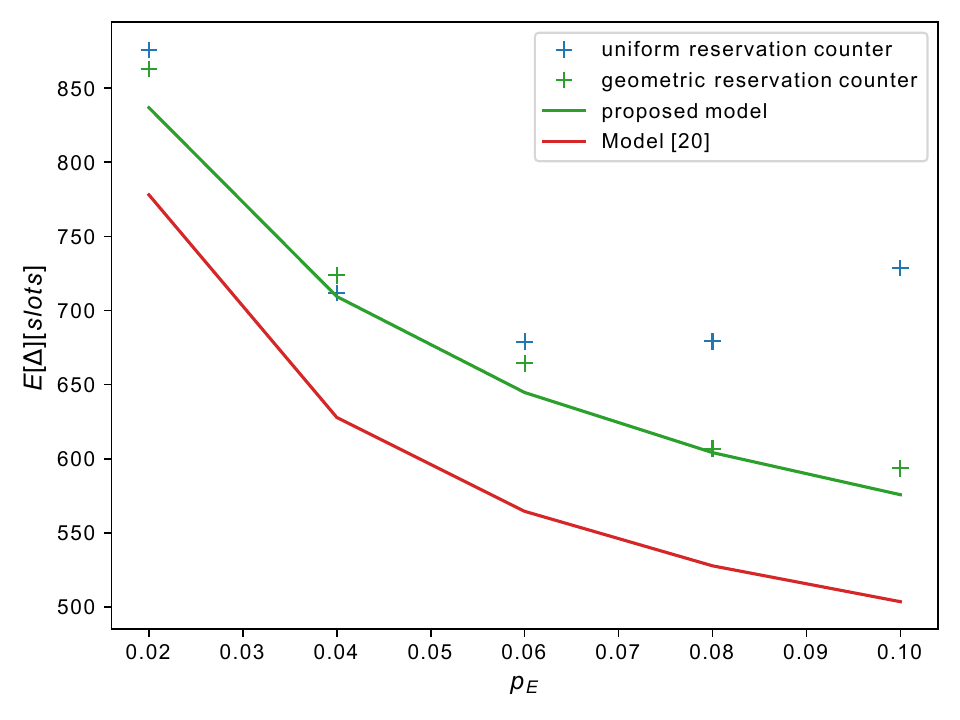}
    \caption{The average AoI as a function of $\pReservationEnd$ for a system with 195 nodes and 200 slots per frame. Analytical models versus simulation (cross markers)}
    \label{fig:expected_AoI}
\end{figure}

\begin{figure}
    \centering
    \includegraphics[width=0.9\linewidth]{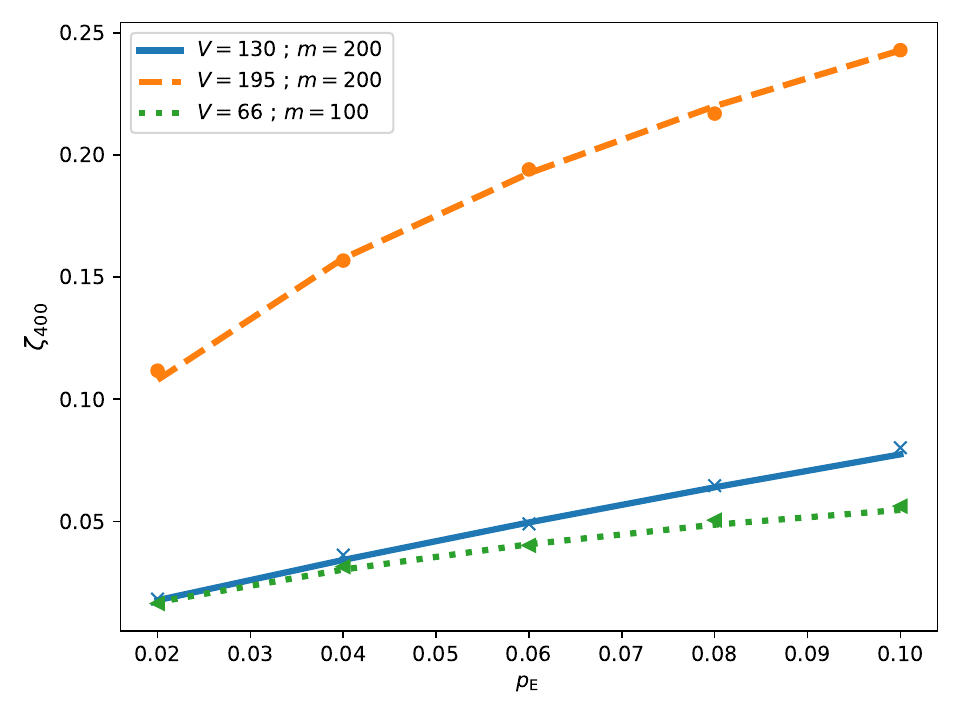}
    \caption{The age violation probability for a threshold $\threshold = 400$ slots as a function of $\pReservationEnd$ and for frame durations of 100 and 200 slots and channel loads of 0.66 and 0.98 nodes/slot, comparing the analytical model (lines) against simulations (markers).}
    \label{fig:outage_AoI}
\end{figure}
\subsubsection*{Age-violation probability} In Fig.~\ref{fig:outage_AoI}, the age-violation probability is illustrated as a function of $\pReservationEnd$ for frame durations of 100 and 200 slots and system sizes of 66, 130 and 195 nodes. For the evaluation, we assume an exemplary application threshold of $\threshold=400$ slots, however, other thresholds in the range of 200 to 2000 slots show the same tendencies. For all parameter combinations, we observe that the age-violation probability increases with higher $\pReservationEnd$ and therefore shorter reservation length. Furthermore, we observe, as expected, a lower age-violation probability for equal load but smaller frame size. This is not surprising: E.g., with a sampling period of 100 slots, two transmissions could be lost without exceeding the AoI threshold. With a sampling rate of 200 slots on the other hand, the threshold is exceeded as soon as a single transmission is lost. 
Due to the increased reservation duration, consecutive collisions are longer lasting but occur less frequently. This means that the case of two collisions that can be compensated for with shorter sampling rates basically no longer occurs. 
Thus, instead of sending twice as many samples per node, we could accommodate almost twice as many nodes at half the sampling rate without increasing the age-violation probability. 

Further, we observe that optimizing for age-violation and average AoI can lead to complementary operational outcomes. Specifically, as system parameters such as $\pReservationEnd$ increase, the age-violation probability tends to worsen even as the average AoI improves. This underlines the necessity for models that capture the full distribution of the \ac{AoI} and associated violation probabilities, beyond mere averages.

Averages often mask the impact of extreme AoI events that can severely degrade performance \cite{abdel-aziz_optimized_2020}. This is particularly significant in the V2X context, where extreme AoI values lead to substantial tracking errors \cite{thunberg_unreliable_2021}, reducing the monitoring capabilities of maneuvering vehicles. This diminishes their ability to respond promptly to unexpected events, thus compromising the safety of driving applications.

These issues are particularly crucial in specific applications like cooperative coordination \cite{thunberg_unreliable_2021}, \ac{CACC} \cite{Vinel_thresholdPAoI_2015}, and platooning \cite{Vinel_thresholdPAoI_2015}. For instance, in cooperative coordination, a failure to promptly communicate abort messages may prevent vehicles from adapting to new road conditions, which can lead to potential collisions at intersections \cite{thunberg_unreliable_2021}.

\section{Assumptions}
\label{sec:assumptions}
In this section, we take a closer look at the assumptions made for the results and discuss the implications. 
Before we discuss the assumptions in detail, we summarize them here: 
(i) The nodes in the system are fully connected. (ii) The system is in steady state. 
(iii) The state $\frameState{\reservationStart{\indexI}{\variableX}}$ at the beginning of a reselection is approximately independent of the state $\frameState{\FrameRecursive{\indexI+1}{\variableX}}$ at the end of the last reselection. 
(iv) For the numerical evaluation, we have also assumed that the number of empty slots in each slot is constant. 

\subsection{Full connectivity}
\label{subsec:full_connectivity}
In Section~\ref{sec:System_Model}, we assume that the system has full connectivity, with collisions being the only source of transmission errors. This assumption is supported by findings from system-level simulations of LTE-V2X channel access, as described in \cite[Fig. 15]{molina-masegosa_comparison_2020}. These simulations demonstrate that within the direct vicinity ($\leq$ 500 m) of a node, collisions are the main cause of packet loss in cellular V2X transmissions while propagation errors play a far smaller role. Hence, for this regime it is reasonable to expect that an analysis of performance metrics such as AoI that is focused on packet loss can still yield a useful approximation of overall communication performance. Examples for communication scenarios in the automotive context in which all nodes are in direct vicinity of each other include cooperative coordination\cite{thunberg_unreliable_2021}, platooning \cite{Vinel_thresholdPAoI_2015}, \ac{CACC}\cite{Vinel_thresholdPAoI_2015} and group start \cite{Xu_V2X_2020}. In all use cases, a group of autonomous or semi-autonomous vehicles coordinates their movement as a group to improve safety and road efficiency by mitigating the effects of unexpected driving actions and maximizing the use of available road space. A detailed description of all use cases can be found in \cite{5gaa_automotive_association_study_2021,5gaa_automotive_association_c-v2x_2021}.

Assuming full connectivity also disregards hidden nodes as an additional source of error. Findings in \cite[Fig. 3-4;6-7]{Rolich_Impact_2023} show that neglecting hidden nodes can sometimes lead to an underestimation of transmission errors and of the average \ac{PAoI}. In the example simulated in \cite{Rolich_Impact_2023}, this leads to a small to moderate error in the estimation of average \ac{PAoI} in the regime of short reservation lengths. As the reservation length increases, this error vanishes, and is negligible even for moderately long reservation lengths. We therefore expect that the approximations reported in this paper are adequate in the regime of medium to long reservation lengths even in the presence of hidden nodes. In future research, the model could be extended to include hidden nodes so that the approximation accuracy can be increased for the regime of very short reservation lengths in the presence of hidden nodes.
%

\subsection{Stationarity}
\label{subsec:stationarity}
The system described in Section~\ref{sec:System_Model} is fully described by the channel occupancy pattern $\selectedSlotVector{\variableX} = \left({\selectedSlot{1}{\variableX},\dots, \selectedSlot{\nodesSystem}{\variableX}}\right)_{\variableX \in \natZero}$. 
Based on the transmission rules in Section~\ref{subsec:transmissionRules}, the transition probability of the channel occupancy pattern from one frame to the next can be stated as
\begin{align}
   \label{def:probSelectedSlot1}
   &\probM{\selectedSlotVector{\variableX+1} = \left(\stateD{1},\dots,\stateD{\nodesSystem}\right)|\selectedSlotVector{\variableX} = \left(\stateC{1},\dots,\stateC{\nodesSystem}\right)}\nonumber\\
   &=\left( 1 - \pReservationEnd \right)^{
       \sum_{\nodeIndexF=1}^{\nodesSystem}\mathds{1}_{\stateD{\nodeIndexF}=\stateC{\nodeIndexF}}
       }
       \left( \frac{\pReservationEnd}{\numEmptySlots{\stateC{1},\dots,\stateC{\nodesSystem}}} \right)^{
       \sum_{\nodeIndexF=1}^{\nodesSystem}\mathds{1}_{\stateD{\nodeIndexF}\neq\stateC{\nodeIndexF}}
       } 
\end{align}  
for all $(\stateC{1},\dots,\stateC{\nodesSystem},\stateD{1},\dots,\stateD{\nodesSystem}) \in \setConnection{\variableX+1}$, 
where 
\begin{multline*}
\setConnection{\variableX+1} := \{ (\stateC{1},\dots,\stateC{\nodesSystem},\stateD{1},\dots,\stateD{\nodesSystem}) \in \{0, \dots, \frameSize-1\}^{2\nodesSystem}:\\ \exists \nodeIndexF,\nodeIndexS \in \{1,\dots,\nodesSystem\} : \stateD{\nodeIndexF}\neq\stateC{\nodeIndexF} \wedge \stateD{\nodeIndexF}=\stateD{\nodeIndexS} \wedge  \nodeIndexS \neq \nodeIndexF \}
\end{multline*}
and as 
\begin{align}
   \label{def:probSelectedSlot2}
   \probM{\selectedSlotVector{\variableX+1} = \left(\stateD{1},\dots,\stateD{\nodesSystem}\right)|\selectedSlotVector{\variableX} = \left(\stateC{1},\dots,\stateC{\nodesSystem}\right)}
   =0
   \end{align}
for all $(\stateC{1},\dots,\stateC{\nodesSystem},\stateD{1},\dots,\stateD{\nodesSystem}) \notin \setConnection{\variableX+1}$.
Therefore, the system is a Markov Chain.
\begin{lemma} 
\label{lem:stationarity}
    The sequence
    $\selectedSlotVector{\variableX} = \left({\selectedSlot{1}{\variableX},\dots, \selectedSlot{\nodesSystem}{\variableX}}\right)_{\variableX \in \natZero}$  is an aperioidc, irreducible Markov chain with state space $\{1,\dots,\frameSize\}^\nodesSystem$\text{.}
    Therefore, the Markov chain has a unique stationary distribution for $\selectedSlotVector{\variableX}$.
\end{lemma}
The proof of Lemma~\ref{lem:stationarity} can be found in Appendix~\ref{sec:proofStationarity}.
Thus, we can make the simplifying assumption that, for sufficiently large $\variableX$, the Markov chain $\selectedSlotVector{\variableX}$ is in its stationary distribution. Lemma~\ref{lem:stationarity} confirms that this assumption holds asymptotically independent of the system’s initial state.
The numerical evaluations we have made for Section~\ref{subsec:numericalResults} suggest that the simulated system reaches its stationary state within maximum 50000 iterations, indicating that this assumption is a reasonable one in practice.
Since the number of empty slots $\numEmptySlots{\variableX}$ is calculated deterministically from the occupancy patterns \eqref{def:emptySlots} and \eqref{def:numEmptySlots}, it has a stationary distribution.

\subsection{Independence}
\label{subsec:AssIndependence}
For better readability of this subsection we introduce abbreviations for the frame state and the number of ongoing transmission at the end of the $\indexI$-th reservation as
\begin{align}
    \label{abrev:frameState}
    \frameStateC{\fRC}{\indexI}&:=\frameState{\FrameRecursive{\indexI}{\variableX}}\\
    \label{abrev:reservationDuration}
    \ongoingTransmissionsC{\fRC}{\indexI}&:=\OngoingTransmissions{\FrameRecursive{\indexI}{\variableX}}.
\end{align}

We assume that, for a node $\nodeIndexF$ making reservations over time, each reservation is independent. In particular, we approximate the probability of the frame state $\frameState{\reservationStart{\rCollidedReselections}{x}}$ being $\rFrameState$ at the beginning of a reservation as consistent and independent of the frame states of previous reservations.

\begin{align}
       \probM{\frameState{\reservationStart{\rCollidedReselections}{x}}=\rFrameState} 
       & \approx \probM{\frameState{\reservationStart{\rCollidedReselections}{x}}=\rFrameState|\frameStateC{\fRC}{\rCollidedReselections+1}\geq 2}  \\
       & \approx \probM{\frameState{\reservationStart{\rCollidedReselections}{x}}=\rFrameState|\frameStateC{\fRC}{\rCollidedReselections+1}=1}        
\end{align}

In reality, there is a dependence between the frame states across reservations. Specifically, the outcome of a previous reservation affects the number of available slots for subsequent reservations. For instance, if a collision occurs, more slots may be available in the next reservation compared to a scenario where no collision occurred.

In order to get an idea of how accurate the approximations are, we therefore compare the probability distribution of the frame state $\frameState{\reservationStart{\rCollidedReselections}{x}}$ at the beginning of a reservation with and without presence of the conditions $\frameStateC{\fRC}{\rCollidedReselections+1} = 1$ and $\frameStateC{\fRC}{\rCollidedReselections+1} \geq 2$ on the frame state at the end of the preceding reservation.

As a measure of similarity, we use the \textit{variational distance} $\variationDistance{P,Q}$, which describes the area between the two probability functions $P$ and $Q$.
In discrete domains, the variational distance is defined as
\begin{equation}
   \variationDistance{P,Q} = 0.5 \sum_{x \in \mathbb{Z}} |P(x)-Q(x)|
\end{equation}   
and assumes values in $[0,1]$, with 0 corresponding to identical distributions. In Fig.~\ref{fig:var_dist_col}, we illustrate the variational distance between $P:=\probM{\frameState{\reservationStart{\rCollidedReselections}{x}}=\rFrameState|\frameStateC{\fRC}{\rCollidedReselections+1}\geq 2}$ and $Q:=\probM{\frameState{\reservationStart{\rCollidedReselections}{x}}=\rFrameState}$ as a function of the channel load $\frac{\nodesSystem}{\frameSize}$ for different numbers of slots per frame for $\pReservationEnd=0.1$. For small frame lengths, the variational distance is comparatively large. For longer frame lengths (min. 50 slots per frame), the variational distance becomes smaller and the approximation therefore more accurate. This can be explained by the observation that the probability $\probM{\frameState{\reservationStart{\rCollidedReselections}{x}}=\rFrameState}$ depends on the number of empty slots in frame $\variableX$. The number of free slots is influenced by a collision in frame $\variableX$ in that there are more free slots available (collided nodes occupy fewer slots). The greater the number of slots in general, the less weight these additional slots have.  
\begin{figure}
    \centering
    \includegraphics[width=0.9\linewidth]{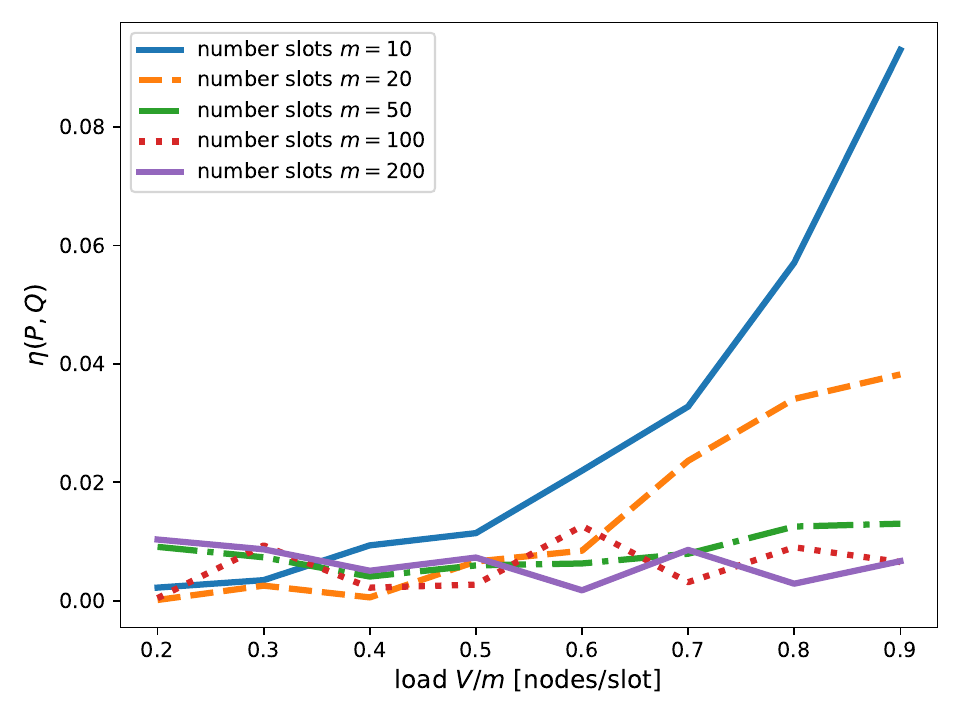}
    \caption{\textbf{Approximation validation.} The variational distance for probabilities $P:= \probM{\frameState{\reservationStart{\rCollidedReselections}{x}}=\rFrameState|\frameState{\FrameRecursive{\rCollidedReselections+1}{x}}\geq 2}$ and $Q:=\probM{\frameState{\reservationStart{\rCollidedReselections}{x}}=\rFrameState}$ is illustrated as function of the system load for different values of the frame length.}
    \label{fig:var_dist_col}
\end{figure}

In Fig.~\ref{fig:var_dist_sin}, we illustrate the variational distance between $P:= \probM{\frameState{\reservationStart{\rCollidedReselections}{x}}=\rFrameState|\frameStateC{\fRC}{\rCollidedReselections+1}=1}$ and $Q:=\probM{\frameState{\reservationStart{\rCollidedReselections}{x}}=\rFrameState}$ as a function of the channel load $\frac{\nodesSystem}{\frameSize}$ for different numbers of slots per frame for $\pReservationEnd=0.1$. The overall low variational distance shows that the approximation is satisfactory. The approximation is particularly accurate for longer frame lengths ($\gtrsim 50$) and low channel loads ($\lesssim 0.7$).    
\begin{figure}
    \centering
    \includegraphics[width=0.9\linewidth]{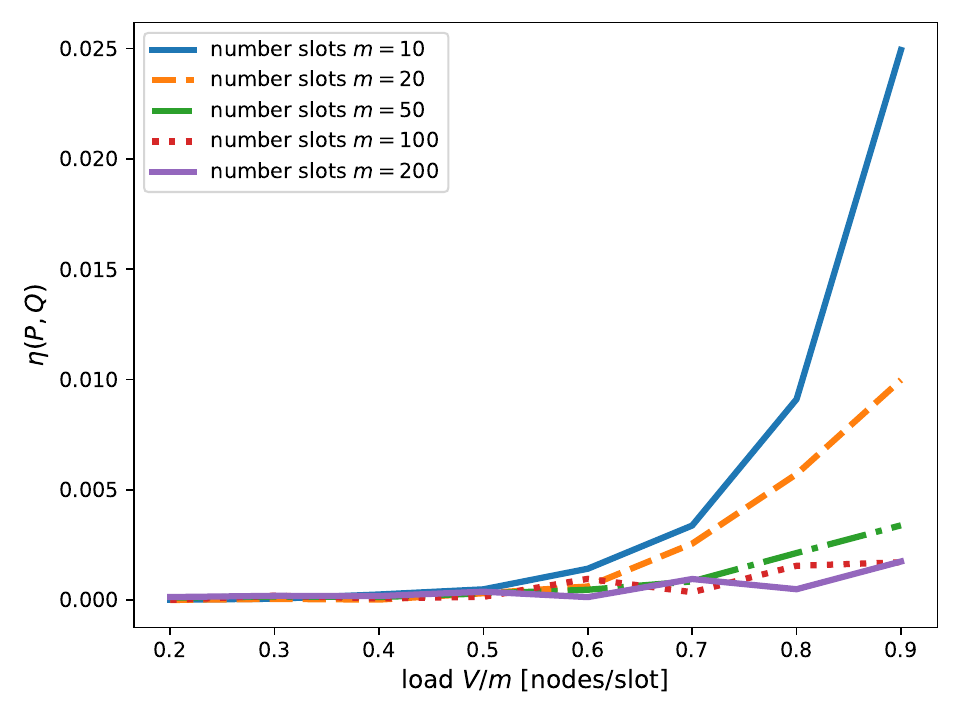}
    \caption{\textbf{Approximation validation.} The variational distance for probabilities $P:= \probM{\frameState{\reservationStart{\rCollidedReselections}{x}}=\rFrameState|\frameState{\FrameRecursive{\rCollidedReselections+1}{x}}=1}$ and $Q:=\probM{\frameState{\reservationStart{\rCollidedReselections}{x}}=\rFrameState}$ is illustrated as function of the system load for different values of the frame length.}
    \label{fig:var_dist_sin}
\end{figure}

With this we come to the conclusion that especially for large $\frameSize$ the distributions are similar enough to use the approximation \eqref{approx:probCollidedReselections}.

Having established an intuitive understanding, we now formalize the approximation. For the proofs throughout the paper, it makes sense to apply the assumption of independence between reservations directly to both the distribution of collided reservations $\collidedReselections{\variableX}$ and the probability of the number of ongoing transmissions reaching $\rOngoingTransmissions$, given a singleton transmission at the end of previous reservations $\probM{\ongoingTransmissionsC{\fRC}{\indexI}=\rOngoingTransmissions|\frameStateC{\fRC}{0}\geq2,\dots, \frameStateC{\fRC}{\rCollidedReselections-1}\geq2,\frameStateC{\fRC}{\rCollidedReselections}=1}$.

Thus, the distribution of collided reservations $\collidedReselections{\variableX}$, as defined in \eqref{def:collidedReselections}, can be approximated as
\begin{align}
\label{approx:probCollidedReselections}
     &\probM{\collidedReselections{\variableX} =\rCollidedReselections}\nonumber\\ 
    \overset{\eqref{def:collidedReselections}}&{=} \probM{\frameStateC{\fRC}{0}\geq2,\dots, \frameStateC{\fRC}{\rCollidedReselections-1}\geq2,\frameStateC{\fRC}{\rCollidedReselections}=1}\nonumber\\
    &\approx \probM{\frameStateC{\fRC}{\rCollidedReselections}=1} \cdot
      \prod_{\indexI=0}^{\rCollidedReselections-1}\probM{\frameStateC{\fRC}{\indexI}\geq2}.
\end{align}
Additionally, we use 
\begin{align}
\label{approx:OngoingTransmissions}
    &\probM{\ongoingTransmissionsC{\fRC}{\indexI}=\rOngoingTransmissions|\frameStateC{\fRC}{0}\geq2,\dots, \frameStateC{\fRC}{\rCollidedReselections-1}\geq2,\frameStateC{\fRC}{\rCollidedReselections}=1}\nonumber\\
    &\approx \probM{\ongoingTransmissionsC{\fRC}{\indexI}=\cdot|
            \frameStateC{\fRC}{\indexI}\geq2}.
\end{align}

To demonstrate the impact of the approximations on the average AoI and AoI violation probability, and therefore the usefulness of these approximations, we compare them with simulations in Section~\ref{subsec:numericalResults}, as shown in Fig.~\ref{fig:pmf_AoI_V195_m200} to Fig.~\ref{fig:outage_AoI}.

\subsection{Number of empty slots}
\label{subSec:emptySlots}

In this paper, we approximate the number of empty slots per frame according to the derivation of the expected value in \cite{Rolich_Impact_2023}. 
\begin{align}
\label{as:emptySlots}
    \numEmptySlots{\selectedSlotVector{\variableX}} \approx  \eNumEmptySlots,
\end{align}
Therefore, we give a detailed account of the derivation in \cite{Rolich_Impact_2023}. In \cite[(2)-(16)]{Rolich_Impact_2023} $\eNumEmptySlots$ is determined iteratively by starting with $\eNumEmptySlots_0=\frac{\frameSize}{\nodesSystem}$ using 
\begin{equation}
    \label{eq:emptySlots}
    \eNumEmptySlots_{\indexI+1}  = \frac{\frameSize}{1+\underline{w}_{\indexI}(\underline{\underline{I}}-\underline{\underline{A}})^{-1}\underline{e}},
\end{equation}
with $\underline{w}_{\indexI}:=(w^1,\dots,w^\nodesSystem)_\indexI$
\begin{equation}
\label{eq:vVec}
      w^{k}_\indexI= \binom{\nodesSystem}{k}\frac{\pReservationEnd}{1+\eNumEmptySlots_\indexI}^{k}
      \left( 1- \frac{\pReservationEnd}{1+\eNumEmptySlots_\indexI}\right)^{\nodesSystem-k}.
\end{equation}
Further, $\underline{\underline{I}}$ is the $\nodesSystem \times \nodesSystem$ identity matrix and $\underline{e}$ denotes a column
vector of ones’s of size $\nodesSystem$. The matrix $\underline{\underline{A}}:=(a_{m,n})_{m,n \in \{0,\dots, \nodesSystem\}}$ is defined as 
\begin{equation}
    a_{m,n}= \begin{cases}
        \binom{m}{n} \pReservationEnd^{n}(1-\pReservationEnd)^{m-n} & m\geq n\\
        0 & else
    \end{cases}
\end{equation}
As we show in Figure \ref{fig:empty_slots} using some example parameters, the distribution of empty slots is concentrated around the expected value. At very high loads, the approximation is less accurate than at low to moderate loads.
   
\begin{figure}
    \centering
    \includegraphics[width=0.9\linewidth]{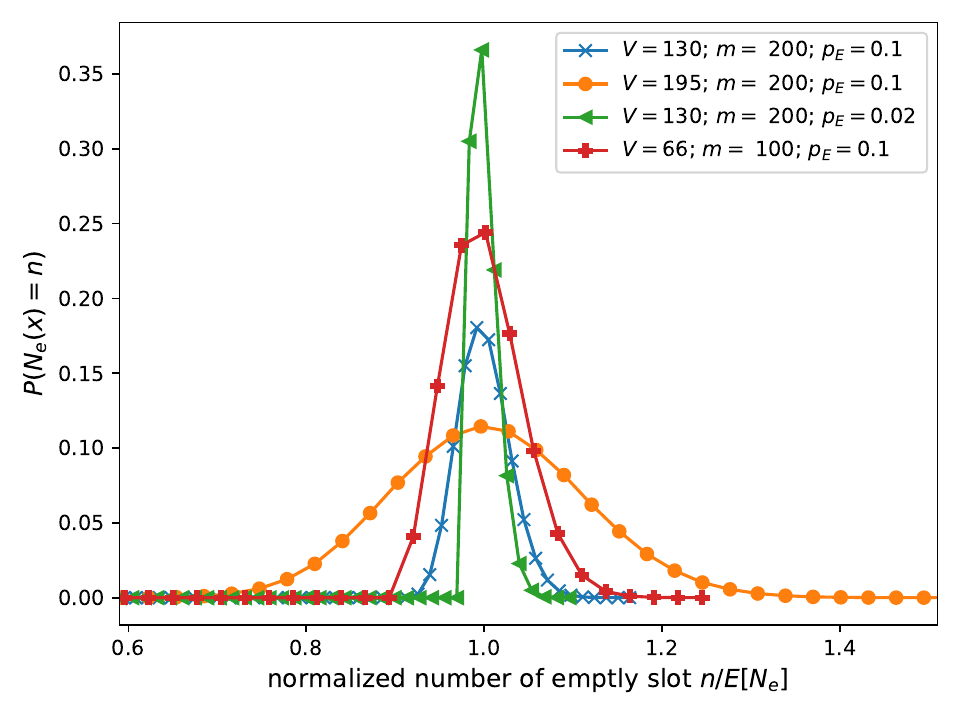}
    \caption{\textbf{Probability distribution empty slots.}}
    \label{fig:empty_slots}
\end{figure}

\section{Conclusion}
\label{sec:conclusion}
In this paper we derived an approximation for the pmf of the AoI for reservation-based channel access schemes, specifically focusing on SPS. We validate the approximation using simulations. With this model, we are able to better quantify the typical ON-OFF behavior which arises from reserving resources periodically, and thereby repeating collisions and successful transmissions alike. In addition, our model offers the first tool for quantifying the age-violation probability, which is important for safety-critical applications. Further, we demonstrate the necessity for cross-layer optimization by showing, for example, that one can accommodate almost twice as many nodes in a system with almost the same age violation probability if the reservation probability and the sampling period are chosen properly. We conclude by emphasizing the critical importance of accurately defining the requirements of the application. As illustrated by the examples, minimizing the average AoI does not guarantee a minimum age violation probability in general. 

This papers also opens up several interesting avenues of further research. This includes cross-layer optimization of communication protocols for specific safety-critical applications, e.g., in autonomous driving, making use of the insights we give on the dependence of performance and reliability on the parameter choice. Another direction would be the development of protocols that use the AoI estimates we propose to dynamically adapt network parameters to changing conditions. Finally, it could also be beneficial to explore extensions to the system model. This could include more detailed modeling for uniform reservation lengths and exploration of the connection between sampling and transmission via related metrics such as stochastic AoI \cite{zhang_statistical_2021,xiao_statistical_2024,Xiao_statistical_2025}.

\begin{appendices}
\label{sec:appendix}
\section{Proof of Theorem~\ref{thm:segmentation}}
\label{sec:proofsSegmentation}
First, we show that the collision duration $\collidedFrames{\variableX}$, compare \eqref{def:collidedFrames}, can be decomposed into reservations. The number of reservations into which we can decompose $\collidedFrames{\variableX}$ is described by the collided reservation counter $\collidedReselections{\variableX}$ as defined in \eqref{def:collidedReselections}. The duration of each reservation is described in terms of the reservation duration $\OngoingTransmissions{\variableX}$ measured from the respective end of the reservation $\FrameRecursive{\rCollidedReselections}{\variableX}$.

\begin{lemma}
\label{lem:collidedFrames}
For all $\variableX \in \natOne$, we can write the collision duration as
\begin{equation}
\label{eq:collidedFrames}
    \collidedFrames{\variableX} = \sum_{\rCollidedReselections=0}^{\collidedReselections{\variableX}-1}\OngoingTransmissions{\FrameRecursive{\rCollidedReselections}{\variableX}},
\end{equation}
where we use the convention that a summation ranging from $0$ to $-1$ has value $0$.
\end{lemma}
\begin{proof}
    If $\collidedReselections{\variableX} = 0$, then by \eqref{def:collidedReselections} we have $\frameState{\variableX} = 1$, so by \eqref{def:collidedFrames}, we have $\collidedFrames{\variableX}=0$ and the lemma holds. Now let $\collidedReselections{\variableX} \in \{1,\dots,\variableX\}$. Then, clearly, the following holds:
    \begin{align*}
        \sum_{\rCollidedReselections=0}^{\collidedReselections{\variableX}-1}\OngoingTransmissions{\FrameRecursive{\rCollidedReselections}{\variableX}}
        \overset{\eqref{def:frameRecursive}}&{=} \sum_{\rCollidedReselections=0}^{\collidedReselections{\variableX}-1}\left( \FrameRecursive{\rCollidedReselections}{\variableX} - \FrameRecursive{\rCollidedReselections+1}{\variableX}\right)\nonumber \\
        \overset{}&{=}\FrameRecursive{0}{\variableX} - \FrameRecursive{\collidedReselections{\variableX}}{\variableX} \nonumber \\
        \overset{(a)}&{=} \variableX - \FrameRecursive{\collidedReselections{\variableX}}{\variableX},
    \end{align*}
    where (a) stands for the definition after \eqref{def:frameRecursive}.
    The proof is complete if we can show that $\variableX - \FrameRecursive{\collidedReselections{\variableX}}{\variableX} = \collidedFrames{\variableX}$. According to \eqref{def:collidedFrames}, this follows if 
    \begin{equation}
    \label{eq:statement1}
        \frameState{\FrameRecursive{\collidedReselections{\variableX}}{\variableX}}=1   
    \end{equation}
    and if for all $j \in \{0,\dots, \variableX -\FrameRecursive{\collidedReselections{\variableX}}{\variableX}-1\}$
    \begin{equation}
        \frameState{\variableX-j}\geq 2.
    \end{equation}
    Setting $i:=\variableX - \FrameRecursive{\collidedReselections{\variableX}}{\variableX} - j$, we equivalently have for all $\indexI \in\{1, \dots,\variableX- \FrameRecursive{\collidedReselections{\variableX}}{\variableX}\}$  
    \begin{equation}
    \label{eq:statement2}
        \frameState{\FrameRecursive{\collidedReselections{\variableX}}{\variableX}+i}\geq 2.
    \end{equation}
   \eqref{eq:statement1} follows from \eqref{def:collidedReselections} and \eqref{eq:statement2} can be argued in two steps. First, according to the transmission rules in Section~\ref{subsec:transmissionRules}, the position $\selectedSlot{}{\reservationStart{\rCollidedReselections}{\variableX}}$ is occupied in every frame from frame $\reservationStart{\rCollidedReselections}{\variableX}$ to frame $\FrameRecursive{\rCollidedReselections}{\variableX}$. That is, no nodes other than those already transmitting at $\selectedSlot{}{\reservationStart{\rCollidedReselections}{\variableX}}$ transmit at this position. Therefore,
   \begin{equation*}
        \frameState{\reservationStart{\rCollidedReselections}{\variableX}} \geq \dots \geq \frameState{\FrameRecursive{\rCollidedReselections}{\variableX}}
   \end{equation*}
    for all $\rCollidedReselections \in \{0,\dots, \variableX\}$.
    Second, by \eqref{def:collidedReselections}, the frame state at the end of a reservation is for all $\rCollidedReselections<\collidedReselections{\variableX}$ 
    \begin{equation*}
        \frameState{\FrameRecursive{\rCollidedReselections}{\variableX}}\geq 2.
    \end{equation*}
    Finally, we note that by the definition of $\reservationStart{\indexN}{\variableX}$ after \eqref{def:frameRecursive}, $\FrameRecursive{\collidedReselections{\variableX}}{\variableX}+1 = \reservationStart{\collidedReselections{\variableX}-1}{\variableX}$. Hence, the term $\FrameRecursive{\collidedReselections{\variableX}}{\variableX}+i$ which appears in \eqref{eq:statement2} is between $\reservationStart{\rCollidedReselections}{\variableX}$ and $\FrameRecursive{\rCollidedReselections}{\variableX}$ with $\rCollidedReselections<\collidedReselections{\variableX}$ for all possible values that $i$ ranges over, which concludes the proof of \eqref{eq:statement2}.
\end{proof}
\begin{proof}[Proof of Theorem~\ref{thm:segmentation}]
    Theorem~\ref{thm:segmentation} now follows by substituting \eqref{eq:collidedFrames} into the definition of the AoI \eqref{def:AoI}.
\end{proof}

\section{Proof of Lemma~\ref{lem:stationarity}}
\label{sec:proofStationarity}
 \begin{proof}
    According to \cite[Theorem 1.8.3]{Norris_1997} an irreducible, aperiodic Markov chain with an invariant distribution has a stationary distribution.
    To prove \textit{irreducibility}, we show that the transition from any $\left(\stateD{1},\dots,\stateD{\nodesSystem}\right)$ to any other state $\left(\stateC{1},\dots,\stateC{\nodesSystem}\right)$ where $\stateC{},\stateD{} \in \{1,\dots,\frameSize\}$ is feasible for all $\nodeIndexF \in \{1,\dots,\nodesSystem\}$ \cite[Proposition 8.4.10]{rosental_2006}. For a system with $\nodesSystem < \frameSize$, as described in Section~\ref{sec:System_Model}, $\{1,\dots,\frameSize\}\backslash\{\stateC{1},\dots,\stateC{\nodesSystem}\}$ is not empty. 
    Let $\tau \in \{1,\dots,\frameSize\}\backslash\{\stateC{1},\dots,\stateC{\nodesSystem}\}$, then according to \eqref{def:probSelectedSlot1}
    \begin{equation*}
        \probM{\selectedSlotVector{\variableX+1} = \left(\tau,\dots,\tau\right)|\selectedSlotVector{\variableX} = \left(\stateC{1},\dots,\stateC{\nodesSystem}\right)}>0
    \end{equation*}
    and
    \begin{equation*}
        \probM{\selectedSlotVector{\variableX+1} = \left(\stateD{1},\dots,\stateD{\nodesSystem}\right)|\selectedSlotVector{\variableX} = \left(\tau,\dots,\tau\right)}>0 \text{.}
    \end{equation*}
    The proof of \textit{aperioidcity} follows from \eqref{def:probSelectedSlot1} as 
    \begin{equation*}
        \probM{\selectedSlotVector{\variableX+1} = \left(\stateD{1},\dots,\stateD{\nodesSystem}\right)|\selectedSlotVector{\variableX} = \left(\stateD{1},\dots,\stateD{\nodesSystem}\right)}>0 \text{.} \qedhere
    \end{equation*}
\end{proof}
\section{Derivation of Approximation~\ref{thm:pmfAoI}}
\label{sec:derivation_pmf}
\subsection{Determining the duration of one collided reservation}
First, we derive all pmfs necessary to specify the pmf of the duration of a collided reservation. We begin with the probability $\probM{\frameState{\reservationStart{i}{\variableX}}=\rFrameState} $ that at the start of the $i-th$ reservation $\reservationStart{i}{\variableX}$ a simultaneous transmission of $\rFrameState$ nodes takes place.
\begin{lemma}
\label{lem:probFrameStateStart}
    For the system in stationary state and all $\pReservationEnd \in [0,1]$, $\variableX \in \natOne$ and $\indexI \in \natZero$, the probability distribution
    $\probM{\frameState{\reservationStart{i}{\variableX}}=\rFrameState} $ is stationary and described as
    \begin{align}
        \label{eq:probFrameStateReselection}
        \probM{\frameState{\reservationStart{i}{\variableX}}=\rFrameState} 
        &=\sum_{\rNumEmptySlots=1}^{\frameSize} 
            \probNumEmptySlots{\rNumEmptySlots} \cdot
            \binomial{\nodesSystem-1}{\frac{\pReservationEnd}{\rNumEmptySlots}}{\rFrameState-1}.
    \end{align}
\end{lemma}
\begin{proof}
    For this proof, we define the number of nodes performing a reselection $\reselectingNodes{\reservationStart{\indexI}{\variableX}}$ in the same frame $\reservationStart{\indexI}{\variableX}$ as node $\nodeIndexF$ as 
    \begin{equation}
       \reselectingNodes{\reservationStart{\indexI}{\variableX}} 
       := \sum_{\substack{\indexJ=1\\\indexJ\neq \nodeIndexF}}^{\nodesSystem}\mathds{1}_{\selectedSlot{\indexJ}{\reservationStart{\indexI}{\variableX}}\neq\selectedSlot{\indexJ}{\reservationStart{\indexI}{\variableX}-1}}.
    \end{equation}
    Clearly, the probability distribution follows a binomial distribution described by
    \begin{equation}
        \probM{\reselectingNodes{\reservationStart{\indexI}{\variableX}}=\rReselectingNodes} = \binomial{\nodesSystem-1}{\pReservationEnd}{\rReselectingNodes}.
    \end{equation}
    Further, according to the transmission rules in Section~\ref{subsec:transmissionRules}, the probability distribution $\probM{
        \frameState{\reservationStart{\indexI}{\variableX}}=\rFrameState|
        \reselectingNodes{\reservationStart{\indexI}{\variableX}}=\rReselectingNodes,
        \numEmptySlots{\FrameRecursive{\indexI+1}{\variableX}} =\rNumEmptySlots
        }$ is also described by a binomial distribution as
    \begin{align}
        &\probM{
        \frameState{\reservationStart{\indexI}{\variableX}}=\rFrameState|
        \reselectingNodes{\reservationStart{\indexI}{\variableX}}=\rReselectingNodes,
        \numEmptySlots{\FrameRecursive{\indexI+1}{\variableX}} =\rNumEmptySlots
        }\nonumber\\ 
        =& \binomial{\rReselectingNodes}{\frac{1}{\rNumEmptySlots}}{\rFrameState-1}.
    \end{align}
    Due to \cite[chapter 2.1., first example]{Pinsky_Introduction_2011}, we can write 
    \begin{align}
    \label{eq:probDoubleBinimial}
        \probM{
        \frameState{\reservationStart{\indexI}{\variableX}}=\rFrameState|
        \numEmptySlots{\FrameRecursive{\indexI+1}{\variableX}} =\rNumEmptySlots
        } 
        = \binomial{\nodesSystem-1}{\frac{\pReservationEnd}{\rNumEmptySlots}}{\rFrameState-1}.
    \end{align}
    Thus, we can write $\probM{\frameState{\reservationStart{i}{\variableX}}=\rFrameState}$ as 
    \begin{align*}
        &\hphantom{{}={}}\probM{\frameState{\reservationStart{i}{\variableX}}=\rFrameState}  \\
        \overset{LTP}&{=} \sum_{\rNumEmptySlots=1}^{\frameSize} 
            \probM{
            \numEmptySlots{\FrameRecursive{\indexI+1}{\variableX}} =\rNumEmptySlots
            } 
            \probM{
            \frameState{\reservationStart{\indexI}{\variableX}}=\rFrameState|
            \numEmptySlots{\FrameRecursive{\indexI+1}{\variableX}} =\rNumEmptySlots}\\
        \overset{\eqref{def:emptySlotsStatinarity}}&{=}
        \sum_{\rNumEmptySlots=1}^{\frameSize} 
            \probNumEmptySlots{\rNumEmptySlots} 
            \probM{
            \frameState{\reservationStart{\indexI}{\variableX}}=\rFrameState|
            \numEmptySlots{\FrameRecursive{\indexI+1}{\variableX}} =\rNumEmptySlots}.     
    \end{align*}
    Combining this with \eqref{eq:probDoubleBinimial} concludes the proof.
\end{proof}
Knowing the state $\rFrameState$ at the start of a reservation $\reservationStart{\indexN}{\variableX}$, we can derive the transition probability $\probM{\frameState{\variableX} \geq \rFrameState| \OngoingTransmissions{\variableX}=\rOngoingTransmissions }$ to a collision in frame $\variableX$ under the condition that the reservation started $\rOngoingTransmissions-1$ frames ago, where the reservation duration $\OngoingTransmissions{\variableX}$ is defined as in \eqref{def:OngoingTransmissions}. 

Clearly, the probability that a node changes position from one frame to the next is given by
\begin{equation}
\label{eq:probFrameReselection}
    \probM{\selectedSlot{}{\variableX-1}\neq\selectedSlot{}{\variableX}} =\pReservationEnd.
\end{equation} 
Therefore, the distribution $\probM{\OngoingTransmissions{\variableX} = \rOngoingTransmissions}$ is stationary.

\begin{lemma}
\label{lem:probFrameStateEnd}
    For the system in stationary state and for all $\rOngoingTransmissions \in \{0,\dots, \variableX\}$, with $\variableX \in \natOne$, the probability distribution $\probM{\frameState{\variableX} \geq 2| \OngoingTransmissions{\variableX}=\rOngoingTransmissions }$ is described as
    \begin{multline}
    \label{eq:probFrameStateEnd}
        \probM{\frameState{\variableX} \geq 2| \OngoingTransmissions{\variableX}=\rOngoingTransmissions }\\
        = 1- \sum_{\rFrameState=1}^{\nodesSystem-1} \left(1-\left(1-\pReservationEnd\right)^{\rOngoingTransmissions-1}\right)^{\rFrameState} \probM{\frameState{\reservationStart{1}{\variableX}}=\rFrameState} 
    \end{multline}
\end{lemma}
  \begin{proof}
For node $\nodeIndexF,\nodeIndexS \in \{1,\dots,\nodesSystem\}$ and $\variableX \in \natOne$, we define the number of frames $\reservationLength{\nodeIndexF}{\variableX}$ in frame $\variableX$ until a new reservation starts for node $\nodeIndexS$ as 
\begin{equation}
\label{def:reservationLength}
    \reservationLength{\nodeIndexS}{\variableX}:=\min\{\rReservationLength\in \natOne : \selectedSlot{}{\variableX+\rReservationLength-1}\neq\selectedSlot{}{\variableX+\rReservationLength}\} \text{.}
\end{equation}
According to \eqref{eq:probFrameReselection}, the cdf of $\reservationLength{\nodeIndexS}{\variableX}$ can be described as 
\begin{equation}
    \label{prob:ReservationLength}
    \probM{\reservationLength{\nodeIndexS}{\variableX} < \rReservationLength} = 1-(1-\pReservationEnd)^{\rReservationLength-1}.
\end{equation}

Further, we use the set $\setCollision{\nodeIndexF}{\variableX}$ to denote all nodes transmitting in the same time slot in frame $\variableX$ as node $\nodeIndexF$. It is defined as
\begin{equation}
    \label{def:setCollision}
    \setCollision{\nodeIndexF}{\variableX}:=\{\nodeIndexS \in \{1,\dots,\nodesSystem\}: \nodeIndexS \neq \nodeIndexF, \selectedSlot{\nodeIndexS}{\variableX}=\selectedSlot{\nodeIndexF}{\variableX}\} \text{.}
\end{equation}
The set of all possible values of $\setCollision{\nodeIndexF}{\variableX}$ is denoted by $\rSetCollision$.

Finally, we can derive the probability for a collided transmission of node $\nodeIndexF$ in frame $\variableX$ given the number of ongoing transmissions $\OngoingTransmissions{\variableX}=\rOngoingTransmissions$ since the last reselection of node $\nodeIndexF$ (compare \eqref{def:OngoingTransmissions}).
\begin{align*}
\label{proof:probCollision}
    &\probM{\frameState{\variableX} \geq 2| \OngoingTransmissions{\variableX}=\rOngoingTransmissions} \\
    \overset{(a)}&{=} 1- \probM{\frameState{\variableX} = 1 | \OngoingTransmissions{\variableX}=\rOngoingTransmissions } \\
    \overset{LTP}&{=} 1- \sum_{s \in \rSetCollision} 
                                \probM{\frameState{\variableX} = 1| \setCollision{\nodeIndexF}{\bar{\variableX}} =s , 
                                    \OngoingTransmissions{\variableX}=\rOngoingTransmissions }
                                \probM{\setCollision{\nodeIndexF}{\bar{\variableX}} =s }\\
    \overset{(b)}&{=} 1- \sum_{s \in \rSetCollision} 
                                \probM{\setCollision{\nodeIndexF}{\bar{\variableX}} =s }
                                \prod_{\nodeIndexS \in s }
                                \probM{\reservationLength{\nodeIndexS}{\bar{\variableX}}< \rOngoingTransmissions}\\
    \overset{\eqref{prob:ReservationLength}}&{=} 1- \sum_{s \in \rSetCollision } 
                                \probM{\setCollision{\nodeIndexF}{\bar{\variableX}} =s }
                                \prod_{\nodeIndexS \in s }
                                \left(1-\left(1-\pReservationEnd\right)^{\rOngoingTransmissions-1}\right)\\
    \overset{}&{=} 1- \sum_{s \in \rSetCollision} 
                                \probM{\setCollision{\nodeIndexF}{\bar{\variableX}} =s }
                                \left(1-\left(1-\pReservationEnd\right)^{\rOngoingTransmissions-1}\right)^{|s |}\\                          
    \overset{}&{=} 1-\sum_{\rFrameState=0}^{\nodesSystem-1} \
                        \probM{|\setCollision{\nodeIndexF}{\bar{\variableX}}| = \rFrameState}
                                \left(1-\left(1-\pReservationEnd\right)^{\rOngoingTransmissions-1}\right)^{\rFrameState}\\
    \overset{(c)}&{=} 1-\sum_{\rFrameState=0}^{\nodesSystem-1} \
                        \probM{\frameState{\bar{\variableX}} = \rFrameState+1}
                                \left(1-\left(1-\pReservationEnd\right)^{\rOngoingTransmissions-1}\right)^{\rFrameState}\\
    \overset{}&{=} 1-\sum_{\rFrameState=1}^{\nodesSystem} \
                        \probM{\frameState{\bar{\variableX}}=\rFrameState}
                                \left(1-\left(1-\pReservationEnd\right)^{\rOngoingTransmissions-1}\right)^{\rFrameState-1}
\end{align*}
where $\bar{\variableX}:=\reservationStart{1}{\variableX}$ 
and (a) is because $\frameState{\variableX}$ can only take values in $\{1,\dots,\nodesSystem\}$. In (b), we leverage the information that, according to the transmission rules in Section~\ref{subsec:transmissionRules}, only the nodes in $\setCollision{\nodeIndexF}{\bar{\variableX}}$ can contribute to a collision in $\variableX$. To determine the probability of a singleton transmission of node $\nodeIndexF$ in $\variableX$ we need to determine the probability that each node in $\setCollision{\nodeIndexF}{\bar{\variableX}}$ has a reservation length below $\rOngoingTransmissions$. As the reservation length of each node is independent of all other nodes, we can write it as a product. For each node, the number of frames until the reservation ends is described by \eqref{def:reservationLength} and is independent of the reservation length of any other. (c) is due to the definition \eqref{def:frameState} and \eqref{def:setCollision}.
\end{proof} 

\begin{lemma}
\label{lem:probOngoingTransmissions}
  For all $\variableX \in \natOne$ and $\rOngoingTransmissions \in \{1,\dots, \variableX\}$ the distribution of $\OngoingTransmissions{\variableX}$ is described as
  \begin{equation}
      \label{def:probOngoingTransmissions}
       \probM{\OngoingTransmissions{\variableX} = \rOngoingTransmissions} = \pReservationEnd \left( 1-\pReservationEnd \right)^{\rOngoingTransmissions}. 
  \end{equation}
\end{lemma}
\begin{proof}
    This clearly follows from \eqref{def:OngoingTransmissions} and \eqref{eq:probFrameReselection}.
\end{proof} 
\subsection{Combining the durations of multiple collided reservations}
So far we have looked at the transition probability within a reservation. To derive the probability distribution of $\collidedFrames{\variableX}$, we need to consider the transition between reservations. To facilitate this analysis, we introduce a simplifying assumption in Section~\ref{subsec:AssIndependence}.

Using this assumption, we approximate the probability distribution of the collision duration $\probM{\collidedFrames{\variableX}=\rCollidedFrames}$.
\begin{approximation}
\label{lem:probCollidedFrames}
The probability distribution in stationary state that node $\nodeIndexF$ in frame $\variableX$ has been in collision for $\rCollidedFrames$ frames can be approximated as 
    \begin{multline}
    \label{prob:collidedFrames}
        \probM{\collidedFrames{\variableX}=\rCollidedFrames}  \\
        \approx 
        \sum_{\rCollidedReselections=0}^{\variableX} 
            \probM{\frameStateC{\fRC}{\rCollidedReselections}=1}
            \bigast_{\indexI=0}^{\rCollidedReselections-1}
             \bigg(\probM{\frameStateC{\fRC}{\indexI}\geq 2| \ongoingTransmissionsC{\fRC}{\indexI}=\cdot} 
            \probM{\ongoingTransmissionsC{\fRC}{\indexI}=\cdot}\bigg)(\rCollidedFrames),
    \end{multline}
using \eqref{abrev:frameState} and \eqref{abrev:reservationDuration}.
\end{approximation}
\begin{proof}[Rationale]
Using the definitions \eqref{abrev:frameState} and \eqref{abrev:reservationDuration}, let $\variableX \in \natOne$, then
    \begin{align*}
         &\probM{\collidedFrames{\variableX}=\rCollidedFrames}\\
         \overset{\eqref{eq:collidedFrames}}&{=} \probM{\sum_{\rCollidedReselections=0}^{\collidedReselections{\variableX}-1}\ongoingTransmissionsC{\fRC}{\rCollidedReselections}=\rCollidedFrames}\\
        \overset{(a)}&{=} \sum_{\rCollidedReselections=0}^{\variableX} \probM{\collidedReselections{\variableX}=\rCollidedReselections} \cdot \probM{{\sum_{\indexI=0}^{\rCollidedReselections-1}\ongoingTransmissionsC{\fRC}{\indexI}=\rCollidedFrames}|\collidedReselections{\variableX}=\rCollidedReselections}\\
        \overset{\eqref{def:collidedReselections}}&{=}
        \begin{aligned}[t]
            \sum_{\rCollidedReselections=0}^{\variableX} 
            &\probM{\collidedReselections{\variableX}=\rCollidedReselections} \\
            &\cdot  
            \probM{\sum_{\indexI=0}^{\rCollidedReselections-1}\ongoingTransmissionsC{\fRC}{\indexI}=\rCollidedFrames|\frameStateC{\fRC}{0}\geq2,\dots, \frameStateC{\fRC}{\rCollidedReselections-1}\geq2,\frameStateC{\fRC}{\rCollidedReselections}=1}
        \end{aligned}
    \\
        \overset{(b)}&{=}
        \begin{aligned}[t]
           \sum_{\rCollidedReselections=0}^{\variableX} 
            &\probM{\collidedReselections{\variableX}=\rCollidedReselections} \\
            &{\cdot}    
            \bigast_{\indexI=0}^{\rCollidedReselections-1}\probM{\ongoingTransmissionsC{\fRC}{\indexI}=\cdot|\frameStateC{\fRC}{0}\geq2,\dots, \frameStateC{\fRC}{\rCollidedReselections-1}\geq2,\frameStateC{\fRC}{\rCollidedReselections}=1}(\rCollidedFrames)
        \end{aligned}
        \\
        \overset{\eqref{approx:OngoingTransmissions}}&{\approx}
            \sum_{\rCollidedReselections=0}^{\variableX} 
            \probM{\collidedReselections{\variableX}=\rCollidedReselections}\cdot\bigast_{\indexI=0}^{\rCollidedReselections-1}\probM{\ongoingTransmissionsC{\fRC}{\indexI}=\cdot|
            \frameStateC{\fRC}{\indexI}\geq2}(\rCollidedFrames) \displaybreak[0]\\
        \overset{\eqref{approx:probCollidedReselections}}&{\approx}
        \begin{aligned}[t]
            \sum_{\rCollidedReselections=0}^{\variableX} &\probM{\frameStateC{\fRC}{\rCollidedReselections}=1}
            \prod_{\indexI=0}^{\rCollidedReselections-1}\probM{\frameStateC{\fRC}{\indexI}\geq2}\\
            &{\cdot}    
            \bigast_{\indexI=0}^{\rCollidedReselections-1}\probM{\ongoingTransmissionsC{\fRC}{\indexI}=\cdot|
            \frameStateC{\fRC}{\indexI}\geq2}(\rCollidedFrames)
        \end{aligned}
        \displaybreak[0] \\
        \overset{}&{=}
            \sum_{\rCollidedReselections=0}^{\variableX} 
            \probM{\frameStateC{\fRC}{\rCollidedReselections}=1}
            \cdot\bigast_{\indexI=0}^{\rCollidedReselections-1}\probM{\ongoingTransmissionsC{\fRC}{\indexI}=\cdot,
            \frameStateC{\fRC}{\indexI}\geq2}(\rCollidedFrames)\\   
        \overset{}&{=} \sum_{\rCollidedReselections=0}^{\variableX}
            \probM{\frameStateC{\fRC}{\rCollidedReselections}=1}
            \bigast_{\indexI=0}^{\rCollidedReselections-1}
            \bigg(\probM{\frameStateC{\fRC}{\indexI}\geq 2| \ongoingTransmissionsC{\fRC}{\indexI}=\cdot} 
            \probM{\ongoingTransmissionsC{\fRC}{\indexI}=\cdot}\bigg)(\rCollidedFrames)          
    \end{align*}
    where (a) is due to the law of total probability, (b) is due to the fact that the summands each refer to reservation durations of different reservations \eqref{def:frameRecursive}, so they are independent of each other \eqref{def:OngoingTransmissions} under the specified condition and a convolution is possible. 
\end{proof}

\subsection{The influence of the time of transmission on the AoI}
In the following lemma, we derive the distribution of $\selectedSlot{}{\variableX}$.
\begin{lemma}
\label{lem:probSelectedSlot}
For the system in stationary state, the probability distribution of $\selectedSlot{}{\variableX}$ is described as
    \begin{equation}
        \label{prob:selectedSlot}
        \probM{\selectedSlot{}{\variableX}=\rSelectedSlot } = \begin{cases}
            \frac{1}{\frameSize}  &  0\leq \rSelectedSlot \leq \frameSize-1 \\
            0  & \text{otherwise}
        \end{cases}
    \end{equation} 
    and hence
    \begin{equation}
        \label{prob:selectedSlotCDF}
        \probM{\selectedSlot{}{\variableX}<\rSelectedSlot } = 
        \begin{cases}
        0 &  \rSelectedSlot < 0\\
        \frac{\rSelectedSlot}{\frameSize} & 0\leq\rSelectedSlot\leq\frameSize-1 \\
        1 & \rSelectedSlot > \frameSize-1.\\
        \end{cases}
    \end{equation} 
\end{lemma}

\begin{proof}
From \eqref{def:probSelectedSlot1} and \eqref{def:probSelectedSlot2}, it is clear that
\begin{align*}
&\hphantom{{}={}}
\probM{
    \selectedSlotVector{\variableX+1} = \left(\stateD{1},\dots,\stateD{\nodesSystem}\right)
    |
    \selectedSlotVector{\variableX} = \left(\stateC{1},\dots,\stateC{\nodesSystem}\right)
}
\\
&=
\probM{
    \selectedSlotVector{\variableX+1} = \left(\permSub{\stateD{}}{1},\dots,\permSub{\stateD{}}{\nodesSystem}\right)
    |
    \selectedSlotVector{\variableX} = \left(\permSub{\stateC{}}{1},\dots,\permSub{\stateC{}}{\nodesSystem}\right)
}
\end{align*}
for every permutation $\permutation$ on $\{0, \dots, \frameSize-1\}$ with $\tau_\permutation:= \permutation(\tau)$. Since $\selectedSlotVector{\variableX}$ has a unique equilibrium by Lemma~\ref{lem:stationarity}, this implies
\begin{equation}
\label{prof:equilibrium}
    \probM{
    \selectedSlotVector{\variableX} = \left(\stateC{1},\dots,\stateC{\nodesSystem}\right)
}
=
\probM{
    \selectedSlotVector{\variableX} = \left(\permSub{\stateC{}}{1},\dots,\permSub{\stateC{}}{\nodesSystem}\right)
}
\end{equation}
whenever the system is in its stationary state. Now let $\stateC{},\stateD{} \in \{0, \dots, \frameSize-1\}$ be arbitrary and $\permutation$ be the transposition of $\stateC{}$ and $\stateD{}$, i.e.,
\begin{equation*}
    \permutation(\tau):=\begin{cases}
        \stateC{}, & \tau=\stateD{}\\
        \stateD{}, & \tau=\stateC{} \\
        \tau, & \text{otherwise.}
        \end{cases}   
\end{equation*}
Then by the definition of $\selectedSlot{}{\variableX}$ with 
\begin{equation*}
    \cOpp{-\nodeIndexF} = (\stateC{1},\dots,\stateC{\nodeIndexF-1},\stateC{\nodeIndexF+1},\dots, \stateC{\nodesSystem})
\end{equation*} and $\setOpp := \{1,\dots,\frameSize-1\}^{\nodesSystem-1}$, it follows that
\begin{align*}
&\hphantom{{}={}}
\probM{\selectedSlot{}{\variableX} = \stateC{}}
\\
&=
\sum_{\cOpp{-\nodeIndexF}\in \setOpp}
    \probM{
        \selectedSlotVector{\variableX} = \left(\stateC{1},\dots,\stateC{\nodeIndexF-1},\stateC{},\stateC{\nodeIndexF+1},\stateC{\nodesSystem}\right)
    }
\\
\overset{\eqref{prof:equilibrium}}&{=}
\sum_{\cOpp{-\nodeIndexF}\in \setOpp}
    \probM{
        \selectedSlotVector{\variableX} = \left(\permSub{\stateC{}}{1},\dots,\permSub{\stateC{}}{\nodeIndexF-1},\stateD{},\permSub{\stateC{}}{\nodeIndexF+1},\permSub{\stateC{}}{\nodesSystem}\right)
    }
\\
\overset{(a)}&{=}
\sum_{\cOpp{-\nodeIndexF}\in \setOpp}
    \probM{
        \selectedSlotVector{\variableX} = \left(\stateC{1},\dots,\stateC{\nodeIndexF-1},\stateD{},\stateC{\nodeIndexF+1},\stateC{\nodesSystem}\right)
    }
\\
&=
\probM{\selectedSlot{}{\variableX} = \stateD{}}
\end{align*}
where step (a) is by reordering the summands. Since this holds for all $\stateC{}, \stateD{} \in \{0, \dots, \frameSize-1\}$, we can conclude that $\selectedSlot{}{\variableX}$ is uniformly distributed, which proves the lemma.
\end{proof}
Taking into account the case distinction in \eqref{def:AoI}, we can finally specify the pmf of the AoI.

 \begin{proof}[Rationale for Approximation~\ref{thm:pmfAoI}]
\eqref{eq:thmPmfAoI2} is clear by \eqref{def:AoI}. 
For \eqref{eq:thmPmfAoI1}, the following holds:
\begin{align*}
    \overset{\phantom{\text{\scriptsize LTP}}}&{\phantom{=}}\probM{\AoI{\variableT}=\rAoI}  \\
    \overset{\text{LTP}}&{=}
        \probM{\selectedSlot{}{\frameIndex{\variableT}}<\position{\variableT}}
        \probM{\AoI{\variableT}=\rAoI| \selectedSlot{}{\frameIndex{\variableT}}<\position{\variableT}}  \\
        \overset{}&{+}\probM{\selectedSlot{}{\frameIndex{\variableT}}\geq\position{\variableT}}
        \probM{\AoI{\variableT}=\rAoI| \selectedSlot{}{\frameIndex{\variableT}}\geq\position{\variableT}} \\
    \overset{\eqref{prob:selectedSlotCDF}}&{=}   
        \frac{\position{\variableT}}{\frameSize}
        \probM{\AoI{\variableT}=\rAoI| \selectedSlot{}{\frameIndex{\variableT}}<\position{\variableT}}  \\
        \overset{}&{+}\frac{\frameSize-\position{\variableT}}{\frameSize}
        \probM{\AoI{\variableT}=\rAoI| \selectedSlot{}{\frameIndex{\variableT}}\geq\position{\variableT}} \\
        \overset{(a)}&{=}\frac{\position{\variableT}}{\frameSize}
        \probM{\collidedFrames{\frameIndex{\variableT}-1}=\frameIndex{\rAoI}-1} \\
        \overset{}&{+}\frac{\frameSize-\position{\variableT}}{\frameSize}
        \probM{\collidedFrames{\frameIndex{\variableT}}=\frameIndex{\rAoI}},
\end{align*}
where (a) is due to the case distinction in \eqref{def:AoI} and $\position{\variableT}=\position{\rAoI}$.
Further, using \eqref{prob:collidedFrames}, we can write $\probM{\collidedFrames{\variableX}=\rCollidedFrames}$ as 
\begin{align*}
    \overset{}&{}\probM{\collidedFrames{\variableX}=\rCollidedFrames} \\
    \overset{\eqref{prob:collidedFrames}}&{\approx}
        \sum_{\rCollidedReselections=0}^{\variableX} 
            \probM{\frameStateC{\fRC}{\rCollidedReselections}=1}
            \bigast_{\indexI=0}^{\rCollidedReselections-1}
             \bigg(\probM{\frameStateC{\fRC}{\indexI}\geq 2| \ongoingTransmissionsC{\fRC}{\indexI}=\cdot} 
            \probM{\ongoingTransmissionsC{\fRC}{\indexI}=\cdot}\bigg)(\rCollidedFrames)\\
    \overset{(a)}&{=}
    \begin{aligned}[t]
        \sum_{\rCollidedReselections=0}^{\variableX} 
            &\sum_{\rOngoingTransmissions=0}^{\variableX} 
            \probM{\ongoingTransmissionsC{\fRC}{\rCollidedReselections}=\rOngoingTransmissions}
            \cdot 
            \bigg( 1 - \probM{\frameStateC{\fRC}{\rCollidedReselections}\geq 2| \ongoingTransmissionsC{\fRC}{\rCollidedReselections}=\rOngoingTransmissions} \bigg)
           \\
            &\cdot
            \bigast_{\indexI=0}^{\rCollidedReselections-1}
             \bigg(\probM{\frameStateC{\fRC}{\indexI}\geq 2| \ongoingTransmissionsC{\fRC}{\indexI}=\cdot} 
            \probM{\ongoingTransmissionsC{\fRC}{\indexI}=\cdot}\bigg)(\rCollidedFrames)
    \end{aligned}
        \\
    \overset{}&{=}    
    \begin{aligned}[t]
        \sum_{\rCollidedReselections=0}^{\variableX} 
            &\sum_{\rOngoingTransmissions=0}^{\variableX} 
            \bigg(\probM{\ongoingTransmissionsC{\fRC}{\rCollidedReselections}=\rOngoingTransmissions}
             - 
             \probM{\ongoingTransmissionsC{\fRC}{\rCollidedReselections}=\rOngoingTransmissions}\probM{\frameStateC{\fRC}{\rCollidedReselections}\geq 2| \ongoingTransmissionsC{\fRC}{\rCollidedReselections}=\rOngoingTransmissions}\bigg)
           \\
            &\cdot
            \bigast_{\indexI=0}^{\rCollidedReselections-1}
             \bigg(\probM{\frameStateC{\fRC}{\indexI}\geq 2| \ongoingTransmissionsC{\fRC}{\indexI}=\cdot} 
            \probM{\ongoingTransmissionsC{\fRC}{\indexI}=\cdot}\bigg)(\rCollidedFrames),
    \end{aligned}    
\end{align*}
where (a) is due to the law of total probability. 

For large $\variableX$, the summations over $\rCollidedReselections$ and $\rOngoingTransmissions$ are dominated by the first few summands, so we approximate them by summations that only go up to $\maxCollidedReselections$, respectively $\maxOngoingTransmissions$. This allows us to approximate 
\begin{align*}
    &\hphantom{{}={}}\probM{\collidedFrames{\variableX}=\rCollidedFrames} \\
    \overset{}&{\approx}    
    \begin{aligned}[t]
        \sum_{\rCollidedReselections=0}^{\maxCollidedReselections} 
            &\sum_{\rOngoingTransmissions=0}^{\maxOngoingTransmissions} 
            \bigg(\probM{\ongoingTransmissionsC{\fRC}{\indexI}=\rOngoingTransmissions}
             - 
             \probM{\ongoingTransmissionsC{\fRC}{\indexI}=\rOngoingTransmissions}\probM{\frameStateC{\fRC}{\indexI}\geq 2| \ongoingTransmissionsC{\fRC}{\indexI}=\rOngoingTransmissions}\bigg)
           \\
            &\cdot
            \bigast_{\indexI=0}^{\rCollidedReselections-1}
             \bigg(\probM{\frameStateC{\fRC}{\indexI}\geq 2| \ongoingTransmissionsC{\fRC}{\indexI}=\cdot} 
            \probM{\ongoingTransmissionsC{\fRC}{\indexI}=\cdot}\bigg)(\rCollidedFrames)
    \end{aligned},    
\end{align*}
independently of $\variableX$.
 Finally, substituting $\probM{\ongoingTransmissionsC{\fRC}{\indexI}=\rOngoingTransmissions}$ with \eqref{def:probOngoingTransmissions} and $\probM{\frameStateC{\fRC}{\indexI}\geq 2| \ongoingTransmissionsC{\fRC}{\indexI}=\rOngoingTransmissions}$ with \eqref{eq:probFrameStateEnd} and \eqref{eq:probFrameStateReselection} respectively, it is clear that 
\begin{equation*}
    \probM{\frameStateC{\fRC}{\indexI}\geq 2| \ongoingTransmissionsC{\fRC}{\indexI}=\rOngoingTransmissions} \cdot \probM{\ongoingTransmissionsC{\fRC}{\indexI}=\rOngoingTransmissions} = \functionP{\rOngoingTransmissions}
\end{equation*}
and 
\begin{equation*}
    \probM{\collidedFrames{\variableX}=\rCollidedFrames} \approx \functionQ{\rCollidedFrames},
\end{equation*}
 which concludes the proof.
\end{proof}
\end{appendices}
\bibliographystyle{IEEEtran}
\bibliography{AoI_SPS_TVT.bib}
\end{document}